\documentclass{llncs}
\pdfoutput=1
\usepackage{latexsym,amsfonts}
\usepackage{amsmath}
\usepackage{amssymb}
\usepackage{array}

\usepackage{graphicx,color}
\usepackage{latexsym,amsfonts}
\usepackage{amsmath}
\usepackage{amssymb}
\usepackage{bbm}

\usepackage{algorithm}
\usepackage{algorithmic}

\newcommand{\R}{{\mathbb{R}}}
\newcommand{\N}{{\mathbb{N}}}
\newcommand{\Nbr}{\mathsf{Nbr}}

\usepackage{tikz}
\usepackage{url}

\usetikzlibrary{calc}

\newcounter{lpnumber} 

\newcommand{\Comp}{\mathcal{C}}

\newcommand{\wt}{\mathsf{wt}}

\newtheorem{new-claim}{Claim}
\newtheorem{prop}{Proposition}
\newtheorem{fact}{Fact}

\begin{document}
\title{Popular Matchings and Limits to Tractability\thanks{This paper is a merger of results shown in the arXiv papers \cite{FPZ18,K18-roommates} along with one in \cite{Kav18} and new results.}}
\author{Yuri Faenza\inst{1} \and Telikepalli Kavitha\inst{2} \and Vladlena Powers\inst{1} \and Xingyu Zhang\inst{1}}
\institute{IEOR, Columbia University, New York, USA,\\
	\email{yf2414, vp2342,  xz2464@columbia.edu} \and Tata Institute of Fundamental Research,\\ Mumbai, India, \email{kavitha@tcs.tifr.res.in}}
\maketitle
\pagestyle{plain}

\begin{abstract}
  We consider {\em popular matching} problems in both bipartite and non-bipartite graphs with strict preference lists. It is known that every stable matching is a min-size popular matching. A subclass of max-size popular matchings called
  {\em dominant matchings} has been well-studied in bipartite graphs: they always exist and there is a simple linear time algorithm to find one. 
  
  We show that stable and dominant matchings are the only two tractable subclasses of popular matchings in bipartite graphs; more precisely, we show that it is NP-complete to decide if $G$ admits a popular matching that is neither stable nor dominant. We also show a number of related hardness results, such as (tight) inapproximability of the maximum weight popular matching problem. In non-bipartite graphs, we show a strong negative result: it is NP-hard to decide whether a popular matching exists or not, and the same result holds if we replace \emph{popular} with \emph{dominant}. 
On the positive side, we show the following results in any graph:
\begin{itemize}
\item  we identify a subclass of dominant matchings called {\em strongly dominant} matchings and show a linear time algorithm to decide if a strongly dominant matching exists or not;
\item we show an efficient algorithm to compute a popular matching of minimum cost in a graph with edge costs and bounded treewidth.
\end{itemize}  
  \end{abstract}

\section{Introduction}
\label{intro}

The marriage problem considered by Gale and Shapley~\cite{GS62} is arguably the most relevant two-sided market model, and has been studied and applied in many areas of mathematics, computer science, and economics. The classical model assumes that the input is a complete bipartite graph, and that each node is endowed with a strict preference list over the set of nodes of the opposite color class. The goal is to find a matching that respects a certain concept of fairness called \emph{stability}. An immediate extension deals with incomplete lists, i.e., it assumes that the input graph is bipartite but not complete. In this setting, the problem enjoys strong and elegant structural properties, that lead to fast algorithms for many related optimization problems (a classical reference in this area is~\cite{GI}). 

In order to investigate more realistic scenarios, several extensions of the above Gale-Shapley model have been investigated. On one hand, one can change the structure of the \emph{input}, admitting e.g.~ties in preference lists, or of preference patterns that are given by more complex choice functions, or allow the input graph to be non-bipartite (see e.g.~\cite{Manlove} for a collection of extensions). On the other hand, one can change the requirements of the \emph{output}, i.e., we could ask for a matching that satisfies properties other than stability. For instance, relaxing the stability condition to \emph{popularity} allows us to overcome one of the main drawbacks of stable matchings and this is its size --- the restriction that blocking pairs are forbidden constrains the size of a stable matching; there are simple instances where the size of a stable matching is only half the size of a maximum  matching (note that a stable matching is maximal, so its size is at least half the size of a maximum matching).

Popularity is a natural relaxation of stability: roughly speaking, a matching $M$ is \emph{popular} if the number of nodes that prefer $M$ to any matching $M'$ is at least the number of nodes that prefer $M'$ to $M$. One can show that stable matchings are popular matchings of minimum size, and a maximum size popular matching can be twice as large as a stable
matching. Hence, popularity allows for matchings of larger size while still guaranteeing a certain fairness condition.

Popular matchings (and variations thereof) have been extensively studied in the discrete optimization community, see e.g. \cite{Biro,CK16,HK11,HK17,KMN09,Kav12,Kav16}, but there are still large gaps on what we know on the tractability of optimization problems over the set of popular matchings. Interestingly, all tractability results in popular matchings rely on connections with stable matchings. For instance, Kavitha~\cite{Kav12}  showed that a max-size popular matching can be found efficiently by a combination of the \emph{Gale-Shapley} algorithm and {\em promotion} of nodes rejected once by all neighbors. Cseh and Kavitha~\cite{CK16} showed that a pair of nodes is matched together in some popular matching if and only if this pair is matched together either in some stable matching or in some \emph{dominant} matching. Dominant matchings are a subclass of max-size popular matchings, and these are equivalent (under a simple linear map) to stable matchings in a larger graph. Recently, Kavitha \cite{Kav18} showed that when there are weights on edges, the problem of finding a max-weight popular matching is NP-hard. 

The notion of popular matchings can be immediately extended to non-bipartite graphs. Popular matchings need not always exist in a non-bipartite graph and it was not known if one can efficiently decide if a popular matching exists or not in a given non-bipartite graph. 

\smallskip

\noindent {\bf Our Contribution.} In this paper, we show NP-hardness, inapproximability results, and polynomial-time algorithms for many problems in popular matchings, some of which have been posed as open questions in many papers in the area (see e.g.~\cite{Cseh on popular matchings,CK16,HK17,Kav16,Manlove}).  Our most surprising result is probably the following: it is NP-complete to decide if a bipartite graph has a popular matching that is neither stable nor dominant (see Theorem~\ref{final-thm}).

Stable matchings and dominant matchings always exist in bipartite graphs and there are linear time algorithms to compute these matchings. Recall that a stable matching is a min-size popular matching and a dominant matching is a max-size popular matching:
thus finding a min-size (similarly, max-size) popular matching is easy. We are not aware of any other natural combinatorial optimization problem where finding elements of min-size (similarly, max-size) is easy but to decide whether there exists {\em any} element that is neither min-size nor max-size is NP-hard.

We also deduce other hardness results: it is NP-complete to decide if there exists a popular matching that contains or does not contain two given edges (see Theorem \ref{thm:forced-forb3} and Section~\ref{sec:forbidden-forced}, where this and some other positive and negative results are discussed); unless P$=$NP, the maximum weight popular matching problem with nonnegative costs cannot be approximated better than a factor $1/2$, and this is tight, since a $1/2$-approximation follows by known results (see Theorem~\ref{thr:mwp}). Moving to non-bipartite graphs, we show that the problem of deciding if a popular matching exists is NP-complete\footnote{Very recently, this result also appeared in~\cite{GMSZ18} on arXiv; our results (from \cite{K18-roommates}) were obtained independently and our proofs are different.} (see Theorem~\ref{main-thm}). We also show that the problem stays NP-complete if we replace \emph{popular} with \emph{dominant} (see Theorem~\ref{second-thm}). 

All together, those negative results settle the main open questions in the area, and cast a dark shadow on the tractability of popular matchings. While stable matchings are a tractable subclass of popular matchings in non-bipartite graphs~\cite{Irv85}, the dominant matching problem is NP-hard in non-bipartite graphs, as shown here. The fact that stable matchings and dominant matchings are the only tractable subclasses of popular matchings in bipartite graphs prompts the following question: is there is a non-trivial subclass of dominant matchings that is tractable in {\em all} graphs?

We show the answer to the above question is ``yes'' by identifying a subclass called {\em strongly dominant} matchings  (see Definition~\ref{def:strong-dominant}): in bipartite graphs, these two classes coincide.
We show a simple linear time algorithm for the problem of deciding if a given graph admits a strongly dominant matching or not and to find one, if so.
We also show that a popular matching of minimum cost (with no restriction on the signs of the cost function) in bipartite and non-bipartite graphs can be found efficiently if the treewidth of the input graph is bounded. 

\smallskip

\noindent{\bf Background and Related results.}
Algorithmic questions for popular matchings were 
first studied in the domain of {\em one-sided} preference lists~\cite{AIKM07} in bipartite instances
where it is only nodes one side, also called agents, that have preferences over their neighbors, also called objects.
Popular matchings need not always exist here, however fractional matchings that are popular always exist~\cite{KMN09}.

Popular matchings always exist in a bipartite instance $G$ with two-sided strict preference lists~\cite{Gar75}. Polynomial time
algorithms to compute a max-size popular matching here were given in \cite{HK11,Kav12} and these algorithms always compute dominant matchings. The equivalence between dominant matchings in the given bipartite graph and stable matchings in a larger bipartite graph shown in \cite{CK16} implies a polynomial time algorithm to solve the max-weight popular matching problem in a complete bipartite graph. It was shown in \cite{Kav18} that it is $\mathsf{NP}$-hard to find a max-size popular matching in a non-bipartite graph (even when a stable matching exists) and it was shown in~\cite{HK17} that it is $\mathsf{UGC}$-hard to compute a $\Theta(1)$-approximation of a max-weight popular matching in non-bipartite graphs.

It was very recently shown \cite{Kav-WG18} that given a bipartite graph $G$ along with a parameter $k\in(\mathsf{min},\mathsf{max})$, where 
$\mathsf{min}$ is the size of a stable matching and $\mathsf{max}$ is the size of a dominant matching,
it is NP-hard to decide whether $G$ admits a popular matching of size $k$ or not. Note that our NP-hardness result is a much stronger statement as we show that the problem of deciding whether $G$ admits a popular matching of 
{\em any} intermediate size (rather than a particular size $k$) is NP-hard.

\smallskip

\noindent {\bf Organization of the paper.} Definitions and important properties of stable and popular matchings are given in Section \ref{prelims}. A linear time algorithm for strongly dominant matchings is given in Section~\ref{section4}. Our main gadget construction is given in Section~\ref{sec:hardness}, where we also show that the problem of deciding if a graph admits a popular matching that is neither stable nor dominant is NP-complete. Other hardness (and some related positive) results are given in Section~\ref{sec:consequences}. In Section~\ref{sec:treewidth}, we give an algorithm for finding a popular matching of minimum cost in a graph with bounded treewidth.

\section{Preliminaries}\label{prelims}

\subsection{Definitions}\label{sec:definitions}

Throughout the paper, we will consider problems where our input is a graph $G$, together with a collection of rankings, one per node of $G$, with each node ranking its neighbors in a strict order of preference. We will denote an edge of $G$ between nodes $u$ and $v$ as $(u,v)$ or $uv$. A matching $M$ in $G$ is {\em stable} if there is no \emph{blocking edge} with respect to $M$, i.e. an edge whose both endpoints strictly prefer each other to their respective assignments in $M$. It follows from the classical work of Gale and Shapley~\cite{GS62} that a stable matching always exists when $G$ is bipartite and such a matching can be computed in linear time.

The notion of popularity was introduced by G\"ardenfors~\cite{Gar75} in 1975. 
We say a node $u$ {\em prefers} matching $M$ to matching $M'$ if either (i)~$u$ is matched in $M$
and unmatched in $M'$ or (ii)~$u$ is matched in both $M, M'$ and $u$ prefers $M(u)$ to $M'(u)$, where $M(u)$ is the partner of $u$ in $M$. 
For any two matchings $M$ and $M'$, let $\phi(M,M')$ be the number of nodes that prefer $M$ to $M'$.

\begin{definition}
\label{pop-def}
A matching $M$ is {\em popular} if  $\phi(M,M') \ge  \phi(M',M)$ for every matching $M'$ in $G$, 
i.e., $\Delta(M,M') \ge 0$ where $\Delta(M,M') = \phi(M,M') -  \phi(M',M)$.
\end{definition}

Thus, there is no matching $M'$ that would defeat a popular matching $M$ in an election between $M$ and $M'$, where each node casts a vote
for the matching that it prefers. Since there is no matching where more nodes are {\em better-off} than in a popular matching, a popular matching
can be regarded as a ``globally stable matching''. Equivalently, popular matchings are weak {\em Condorcet winners}~\cite{condorcet} in the voting instance where nodes are voters and all feasible matchings are the candidates.

Though (weak) Condorcet winners need not exist in a general voting instance, popular matchings always exist in bipartite graphs,
since every stable matching is popular~\cite{Gar75}. Popular matchings have been well-studied in bipartite graphs, in particular,
a subclass of max-size popular matchings called {\em dominant matchings} is well-understood~\cite{CK16,HK11,Kav12}.

\begin{definition}
  \label{def:dominant}
  A popular matching $M$ is dominant in $G$ if $M$ is more popular than any larger matching in $G$, i.e.,  $\Delta(M,M') > 0$ for any matching
  $M'$ such that $|M'| > |M|$.
\end{definition}  

Dominant matchings always exist in a bipartite graph and such a matching can be computed in linear time~\cite{Kav12}. Every polynomial time algorithm currently known to find a popular matching in a bipartite graph finds either a stable matching~\cite{GS62} or a dominant matching~\cite{HK11,Kav12,CK16}.

In some problems, together with the graph $G$ and the preference lists, we will also be given a weight function $c: E \rightarrow \R$. The \emph{weight} (or \emph{cost}) of a matching $M$ of $G$ (with respect to $c$) is defined as $c(M):=\sum_{e \in M} c(e)$.

\subsection{Combinatorial characterization of popular and dominant matchings}
\label{sec:comb-prelims}
Fix a matching $M$ of $G$. A node $u$ of $G$ is \emph{$M$-exposed} if $\delta(u) \cap M=\emptyset$, and \emph{$M$-covered} otherwise. An \emph{$M$-alternating} path (resp. cycle) in $G$ is a path (resp. cycle) whose edges alternate between $M$ and in $E(G)\setminus M$. An $M$-alternating path is \emph{$M$-augmenting} if its first and last nodes are $M$-exposed. We can associate \emph{labels} to edges from $E\setminus M$ as follows:
\begin{itemize}
\item an edge $(u,v)$ is $(-,-)$ if both $u$ and $v$ prefer their respective partners in $M$ to each other;
\item $(u,v) = (+,+)$ if $u$ and $v$ prefer each other to their partners in $M$;
\item  $(u,v) = (+,-)$ if $u$ prefers $v$ to its partner in $M$ and $v$ prefers its partner in $M$ to $u$.  
\end{itemize}

We write $(u,v)=(+,+)$ or \emph{$uv$ is a $(+,+)$ edge}, and similarly for the other cases.
  We also say that \emph{the label of $(u,v)$ at $u$ is $+$} (resp. $-$) if $u$ prefers $v$ to its current partner (resp. its current partner to $v$). 
  Note that \emph{blocking edges} introduced in Section \ref{sec:definitions} coincide exactly with $(+,+)$ edges.

  The graph $G_M$ is defined as the subgraph of $G$ obtained by deleting edges that are labeled $(-,-)$, and by attaching to other edges not in $M$ the appropriate labels defined above. Observe that $M$ is also a matching of $G_M$, hence definitions of $M$-alternating path and cycles apply in $G_M$ as well. These definitions can be used to obtain a characterization of popular matchings in terms of forbidden substructures of $G_M$, as shown in~\cite{HK11}.

\begin{theorem}\label{thr:characterize-popular}
A matching $M$ of $G$ is popular if and only if $G_M$ does not contain any of the following:
	\begin{enumerate}
		\item[(i)]\label{it:circuit} an $M$-alternating cycle with a $(+,+)$ edge.
		\item[(ii)]\label{it:two-plusplus}an $M$-alternating path that starts and ends with two distinct $(+,+)$ edges.
		\item[(iii)]\label{it:plusplus-and-unmatched}
		 an $M$-alternating path that starts from an $M$-exposed node and ends with a $(+,+)$ edge. 
	\end{enumerate}
\end{theorem}




Graph $G_M$ can also be used to obtain a characterization of dominant matchings, as shown in~\cite{CK16}. 
\begin{theorem}
\label{thm:dominant}
  Let $M$ be a popular matching. $M$ is dominant if and only if there is no $M$-augmenting path in $G_M$.
\end{theorem}

The above characterizations can be used to efficiently decide if given a matching $M$ is popular (similarly, dominant), see \cite{HK11}. Hence, in most of our NP-completeness reductions, we focus on proving the hardness part.



\smallskip

The matchings that satisfy Definition~\ref{def:dominant} were called ``dominant'' in~\cite{CK16},
however dominant matchings in bipartite graphs were first constructed in \cite{HK11}.
It was observed in \cite{HK11} that Definition~\ref{def:strong-dominant} given below was a sufficient condition for a
matching $M$ to be a max-size popular matching and the goal in  \cite{HK11} was to efficiently construct
such a matching in a bipartite graph. It was shown in \cite{Kav12} that if
$M$ satisfies conditions~(i)-(iv) in Definition~\ref{def:strong-dominant} then $M$ satisfies
the condition given in Theorem~\ref{thm:dominant} along with the conditions given in Theorem~\ref{thr:characterize-popular};
thus $M$ is a dominant matching.


\begin{definition}
  \label{def:strong-dominant}
A matching $M$ is strongly dominant in $G = (V,E)$ if there is a partition $(L,R)$ of the node set $V$ such that 
(i)~$M \subseteq L \times R$, (ii)~$M$ matches all nodes in $R$,
(iii)~every $(+,+)$ edge is in $R \times R$, and (iv)~every edge in $L \times L$ is $(-,-)$.
\end{definition}

Consider the complete graph on 4 nodes $d_0, d_1, d_2, d_3$ where $d_0$'s preference list is $d_1 \succ d_2 \succ d_3$
(i.e., top choice $d_1$, followed by $d_2$ and then $d_3$), $d_1$'s preference list is $d_2 \succ d_3 \succ d_0$,
$d_2$'s preference list is $d_3 \succ d_1 \succ d_0$, and $d_3$'s preference list is $d_1 \succ d_2 \succ d_0$. 
This instance (see Fig.~\ref{D:example}) has no stable matching. The matching
$M = \{(d_0,d_1),(d_2,d_3)\}$ is a strongly dominant matching here with $L = \{d_0,d_2\}$ and $R = \{d_1,d_3\}$.
$M \subseteq L \times R$ and it is a perfect matching. Moreover, the edge $(d_0,d_2) \in L \times L$ is $(-,-)$
and there is only one $(+,+)$ edge here, which is $(d_1,d_3) \in R \times R$. 

\smallskip

In bipartite graphs, every dominant matching is strongly dominant~\cite{CK16}.
However in non-bipartite graphs, not every dominant matching is strongly dominant. For instance, consider the following graph on 4 nodes
$a, b, c, d$ where $a$'s preference list is $b \succ c \succ d$, while $b$'s preference list is $a \succ c$ and
$c$'s preference list is $a \succ b$ and $d$'s only neighbor is $a$. It is simple to check that the matching $\{(a,d),(b,c)\}$
is popular; moreover it is a perfect matching and hence it is dominant. However it is {\em not} strongly dominant as both
$(a,b)$ and $(a,c)$ are $(+,+)$ edges and one of $b,c$ has to be in $L$.

\subsection{Dual certificates for stable, popular, and dominant matchings}
\label{sec:certificates}

Let $\tilde{G}$ be the graph $G$ augmented with {\em self-loops}, i.e., it is assumed that every node is its own last choice.
Corresponding to any matching $N$ in $G$, there is a perfect matching $\tilde{N}$ in $\tilde{G}$ defined as follows:
$\tilde{N} = N \cup \{(u,u): u$ is left unmatched in $N\}$. 

Let $M$ be any matching in $G$. Corresponding to $M$, we can define an edge weight function $\wt_M$ in $\tilde{G}$ as follows.
\begin{equation*} 
\wt_M(u,v) = \begin{cases} 2   & \text{if\ $(u,v)$\ is\ labeled\ $(+,+)$}\\
	                     -2 &  \text{if\ $(u,v)$\ is\ labeled\ $(-,-)$ }\\			
                              0 & \text{otherwise}
\end{cases}
\end{equation*}

Thus $\wt_M(e) = 0$ for every $e \in M$.
We need to define $\wt_M$ on self-loops as well: for any node~$u$, let $\wt_M(u,u) = 0$ if $u$ is unmatched in $M$, else
let $\wt_M(u,u) = -1$. 
It is easy to see that for any matching $N$ in $G$, $\Delta(N,M) = \wt_M(\tilde{N})$,
where $\Delta(N,M) = \phi(N,M) - \phi(M,N)$ (see Definition~\ref{pop-def}).
Thus $M$ is popular if and only if every perfect matching in the graph $\tilde{G}$ has weight at most 0.

\subsubsection{Certificates in bipartite graphs.} Here we give a quick overview of the LP framework of popular matchings in bipartite graphs from \cite{KMN09}
along with some results from \cite{Kav16,Kav18}.

\begin{theorem}[\cite{KMN09}]
\label{thm:witness}
  Let $M$ be any matching in $G = (A \cup B, E)$. The matching $M$ is popular if and only if there exists a vector $\vec{\alpha} \in \mathbb{R}^n$ (where $n = |A \cup B|$) such that $\sum_{u \in A \cup B}\alpha_u = 0$ and
  \begin{eqnarray*}
    \alpha_{a} + \alpha_{b} & \ \ \ge \ \ & \wt_{M}(a,b) \ \ \ \forall\, (a,b)\in E\\
    \alpha_u & \ \ \ge \ \ & \wt_M(u,u) \ \ \ \forall\, u\in A \cup B.
  \end{eqnarray*}  
\end{theorem}  

The vector $\vec{\alpha}$ will be an optimal solution to the LP that is dual to the max-weight perfect matching LP in $\tilde{G}$
(with edge weight function $\wt_M$). For any popular matching $M$, a vector $\vec{\alpha}$ as given in Theorem~\ref{thm:witness}
will be called a {\em witness} to $M$.

A stable matching has the all-zeros vector $\vec{0}$ as a witness while it follows from \cite{CK16} that
a dominant matching $M$ has a witness $\vec{\alpha}$ where
$\alpha_u \in \{\pm 1\}$ for all nodes $u$ matched in $M$ and $\alpha_u = 0$ for all nodes $u$ left unmatched in $M$.
The following lemma will be useful to us. Let $|A\cup B| = n$.

\begin{lemma}[\cite{Kav16}]
  Every popular matching in $G = (A \cup B, E)$ has a witness in $\{0,\pm 1\}^n$.
\end{lemma}

Call $s\in V$ a {\em stable node} if it is matched in some (equivalently, all) stable matching(s)~\cite{GS85}. Every popular matching has to match all stable nodes~\cite{HK11}. A node that is not stable is called an {\em unstable node}.

Call any $e \in E$ a {\em popular edge} if there is some popular matching in $G$ that contains $e$.
Let $M$ be a popular matching in $G = (A \cup B, E)$ and let $\vec{\alpha} \in \{0,\pm 1\}^n$ be a witness of $M$.
Lemma~\ref{prop0} given below follows from complementary slackness conditions.

\begin{lemma}[\cite{Kav18}]
\label{prop0}
For any popular edge $(a,b)$, we have $\alpha_a + \alpha_b = \wt_M(a,b)$.
For any unstable node $u$ in $G$, if $u$ is left unmatched in $M$, then $\alpha_u = 0$ else $\alpha_u = -1$.
\end{lemma}

The popular  subgraph $F_G$ is a useful subgraph of $G$ defined in \cite{Kav18}.

\begin{definition}
\label{def:popular-subgraph}
  The {\em popular subgraph} $F_G = (A \cup B, E_F)$ is the subgraph of $G = (A \cup B, E)$
whose edge set $E_F$ is the set of popular edges in $E$.
\end{definition}

The graph $F_G$ need not be connected. Let $\Comp_1,\ldots,\Comp_h$ be the various components in $F_G$. 
Recall that $M$ is a popular matching in $G$ and $\vec{\alpha} \in \{0,\pm 1\}^n$ is a witness of $M$.

\begin{lemma}[\cite{Kav18}]
\label{prop1}
For any connected component $\Comp_i$ in $F_G$, either $\alpha_u = 0$ for all nodes $u \in \Comp_i$ or 
$\alpha_u \in \{\pm 1\}$ for all nodes $u \in \Comp_i$. Moreover, if $\Comp_i$ contains one or 
more unstable nodes, either all these unstable nodes are matched in $M$ or none of them is 
matched in $M$.
\end{lemma}

The following definition marks the state of each connected component  $\Comp_i$ in $F_G$ as ``zero'' or ``unit'' in $\vec{\alpha}$
--- this classification will be useful to us in our hardness reduction.

\begin{definition}
  \label{def:stab-domn}
A connected component $\Comp_i$ in $F_G$ is in {\em zero state} in $\vec{\alpha}$ if  $\alpha_u = 0$ for all nodes $u \in \Comp_i$.
Similarly, $\Comp_i$ in $F_G$ is in {\em unit state} in $\vec{\alpha}$ if  $\alpha_u \in \{\pm 1\}$ for all nodes $u \in \Comp_i$.
\end{definition}

\subsubsection{Certificates in non-bipartite graphs.}
The following theorem shows that the sufficient condition in Theorem~\ref{thm:witness} certifies popularity in non-bipartite graphs as well.

\begin{theorem}
\label{thm:non-bipartite}
Let $M$ be any matching in $G = (V, E)$. The matching $M$ is popular if there exists $\vec{\alpha} \in \mathbb{R}^{|V|}$ such that
$\sum_{u \in V} \alpha_u = 0$ and
  \begin{eqnarray*}
    \alpha_u + \alpha_v & \ \ \ge \ \ & \wt_{M}(u,v) \ \ \ \forall\, (u,v)\in E\\
    \alpha_u & \ \ \ge \ \ & \wt_M(u,u) \ \ \ \forall\, u\in V.
  \end{eqnarray*}  
\end{theorem}  

The proof of Theorem~\ref{thm:non-bipartite} follows by considering the max-weight perfect matching LP in the graph $\tilde{G}$ with edge weight
function $\wt_M$ as the primal LP.
It is easy to see that if there exists a vector $\vec{\alpha} \in \mathbb{R}^{|V|}$ as given above then the optimal value of the dual LP is at most 0,
equivalently, $\wt_M(\tilde{N}) \le 0$ for all matchings $N$ in $G$, i.e., $M$ is a popular matching.

If $M$ is a popular matching that admits $\vec{\alpha} \in \mathbb{R}^{|V|}$ satisfying the conditions in Theorem~\ref{thm:non-bipartite},
we will call $\vec{\alpha}$ a  {\em witness} of $M$.

Note that any stable matching in $G$ has $\vec{0}$ as a witness. 
The witness of the matching $M$ described in Section~\ref{sec:comb-prelims}
on the nodes $d_0,d_1,d_2,d_3$ is $\vec{\alpha}$ where $\alpha_{d_1} = \alpha_{d_3} = 1$ and $\alpha_{d_0} = \alpha_{d_2} = -1$.
Consider the popular (but not strongly dominant) matching $M = \{(a,d),(b,c)\}$ in the other instance described in Section~\ref{sec:comb-prelims}:
this matching $M$ does not admit any witness as given in Theorem~\ref{thm:non-bipartite}.

We will show in Section~\ref{section4} that every strongly dominant matching $M$ admits a witness $\vec{\alpha}$ as given in Theorem~\ref{thm:non-bipartite};
moreover, there will be a witness $\vec{\alpha}$ such that for every node $u$ matched in $M$, $\alpha_u = \pm 1$.

\section{Strongly dominant matchings}
\label{section4}
In this section we generalize the max-size popular matching algorithm for bipartite graphs~\cite{Kav12} to solve the strongly dominant matching 
problem in all graphs. We show a surprisingly simple reduction from the strongly dominant matching
problem in $G = (V,E)$ to the stable matching problem in a new graph $G' = (V,E')$. Thus Irving's algorithm~\cite{Irv85},
which efficiently solves the stable matching problem in all graphs, when run in $G'$, solves our problem.

The graph $G'$ can be visualized as the bidirected graph corresponding to $G$. The node set of $G'$ is the same as that of $G$.
For every $(u,v) \in E$, there will be 2 edges in $G'$ between $u$ and $v$: one directed from $u$ to $v$ which will be denoted by $(u^+,v^-)$
or $(v^-,u^+)$ and the other directed from $v$ to $u$ which will be denoted by $(u^-,v^+)$ or $(v^+,u^-)$.

For any $u \in V$, if $u$'s preference list in $G$ is $v_1 \succ v_2 \succ \cdots \succ v_k$ then $u$'s
preference list in $G'$ is $v^-_1 \succ v^-_2 \succ \cdots \succ v^-_k \succ v^+_1 \succ v^+_2 \succ \cdots \succ v^+_k$.
The neighbor $v_i^-$ corresponds to the edge $(u^+,v_i^-)$ and the neighbor $v_i^+$ corresponds to the edge $(u^-,v_i^+)$.
Thus $u$ prefers {\em outgoing} edges to {\em incoming} edges: among outgoing edges (similarly, incoming edges), its order is its original
preference order.

			

\begin{itemize}
\item A matching $M'$ in $G'$ is a subset of $E'$ such that for each $u \in V$, $M'$ contains at most one edge incident to $u$, i.e.,
     at most one edge in $\{(u^+,v^-),(u^-,v^+): v \in \Nbr(u)\}$ is in $M'$, where $\Nbr(u)$ is the set of $u$'s neighbors in $G$.
\item  For any matching $M'$ in $G'$, define the {\em projection} $M$ of $M'$ as follows:

\[M = \{(u,v): (u^+,v^-)\ \mathrm{or}\ (u^-,v^+)\ \mathrm{is\ in}\ M'\}.\]

It is easy to see that $M$ is a matching in $G$.
\end{itemize}

\begin{definition}
  A matching $M'$ is stable in $G'$ if for every edge $(u^+,v^-) \in E'\setminus M'$: either (i)~$u$ is matched in $M'$
  to a neighbor ranked better than $v^-$ or (ii)~$v$ is matched in $M'$ to a neighbor ranked better than $u^+$.
\end{definition}

We now present our algorithm to find a strongly dominant matching in $G = (V,E)$. 

\begin{enumerate}
\item Build the bidirected graph $G' = (V,E')$. 
\item Run Irving's stable matching algorithm in $G'$.
\item If a stable matching $M'$ is found in $G'$ then
  return the projection $M$ of $M'$.
  
  Else return ``$G$ has no strongly dominant matching''.
\end{enumerate}  

Note that
running Irving's algorithm in the bidirected graph $G'$ is exactly the same as running Irving's algorithm in the
simple undirected graph $H$ that has {\em three} copies of each node $u \in V$:
these are $u^+, u^-$, and $d(u)$. Corresponding to each edge $(u,v)$ in $G$, there
will be the two edges $(u^+,v^-)$ and $(u^-,v^+)$ in $H$ and for each $u \in V$, the graph $H$ also has the edges
$(u^+,d(u))$ and $(u^-,d(u))$.

If $u$'s preference list in $G$ is $v_1 \succ v_2 \succ \cdots \succ v_k$ then $u^+$'s
preference list in $H$ will be $v^-_1 \succ v^-_2 \succ \cdots \succ v^-_k \succ d(u)$ and $u^-$'s
preference list will be $d(u) \succ v^+_1 \succ v^+_2 \succ \cdots \succ v^+_k$.
The preference list of $d(u)$ will be $u^+ \succ u^-$. Thus in any stable matching in $H$, one of
$u^+,u^-$ has to be matched to $d(u)$.

Before we prove the correctness of our algorithm, we will characterize strongly
dominant matchings in terms of their witnesses.

\begin{theorem}
  \label{lem:strongly-dominant}
  A matching $M$ is strongly dominant in $G$ if and only if there exists $\vec{\alpha}$ that satisfies Theorem~\ref{thm:non-bipartite} 
  such that $\alpha_u = \pm 1$ for all nodes $u$ matched in $M$ and $\alpha_u = 0$ for all $u$ unmatched in~$M$.
\end{theorem}  
\begin{proof}
Let $M$ be a strongly dominant matching in $G = (V,E)$. So $V$ can be partitioned into $L \cup R$ such that properties~(i)-(iv) in
Definition~\ref{def:strong-dominant} are satisfied. 
We will construct $\vec{\alpha}$ as follows. For $u \in V$:
\begin{itemize}
\item if $u \in R$ then set $\alpha_u = 1$
\item else set $\alpha_u = -1$ for $u$ matched in $M$ and $\alpha_u = 0$ for  $u$ unmatched in $M$.
\end{itemize}

Since $M$ matches all nodes in $R$, all nodes unmatched in $M$ are in $L$. Thus $\alpha_u = 0$ for all $u$ unmatched in $M$
and $\alpha_u = \pm 1$ for all $u$ matched in $M$. For any edge $(u,v) \in M$, since $M \subseteq L \times R$, 
$(\alpha_u,\alpha_v) \in \{(1,-1),(-1,1)\}$ and so $\alpha_u + \alpha_v = 0$. Thus $\sum_{u \in V}\alpha_u = 0$.

We will now show that $\vec{\alpha}$ satisfies Theorem~\ref{thm:non-bipartite}.
We claim that $\alpha_u \ge \wt_M(u,u)$.
This is because $\alpha_u = 0 = \wt_M(u,u)$ for $u$ left unmatched in $M$ and $\alpha_u \ge -1 = \wt_M(u,u)$ for $u$ matched in $M$. 
We will now show that $\alpha_u + \alpha_v \ge \wt_M(u,v)$ for any edge $(u,v)$ in $G$, i.e., the edge $(u,v)$ is {\em covered}.
Recall that $\wt_M(u,v) \in \{0, \pm 2\}$.
\begin{itemize}
\item Since $\wt_M(e) \le 2$ for any edge $e$ and $\alpha_u = 1$ for all $u \in R$, all edges in $R \times R$ are covered. 
\item We also know that any edge in $L \times L$ is labeled $(-,-)$, i.e., $\wt_M(u,v) = -2$ for any edge 
$(u,v) \in L \times L$. Since $\alpha_u \ge -1$ for any $u \in L$, edges in $L \times L$ are covered. 
\item We also know that all $(+,+)$ edges are in $R \times R$ and so $\wt_M(u,v) \le 0$ for all 
$(u,v) \in L \times R$. Since $\alpha_u \ge -1$ for $u \in L$ and $\alpha_v = 1$ for $v \in R$,  all edges in $L \times R$ are covered.
\end{itemize}

We will now show the converse. Let $M$ be a matching with a witness $\vec{\alpha}$ as given in the statement of the theorem.
The matching $M$ is popular (by Theorem~\ref{thm:non-bipartite}).
Note that we can interpret $\vec{\alpha}$ as the optimal solution to the dual LP of the maximum weight perfect matching LP in $\tilde{G}$ with weights given by $\wt_M(u,v)$, of which $\tilde M$ is an optimal solution (since $M$ is popular).
Hence, the pair $(\tilde M,\vec{\alpha})$ satisfy complementary slackness conditions.

In order to show $M$ is strongly dominant, we will obtain a partition $(L, R)$ of $V$ as follows: let
$R = \{u: \alpha_u = 1\}$ and $L = \{u: \alpha_u \ \text{is\ either}\ 0 \ \text{or -1}\}$.
Complementary slackness conditions on the dual LP imply that if $(u,v) \in M$ 
then $\alpha_u + \alpha_v = \wt_M(u,v) = 0$. Since $u,v$ are matched, $\alpha_u,\alpha_v \in \{\pm 1\}$; so one of $u,v$ is in $L$ and the
other is in $R$. Thus $M \subseteq L \times R$.

We have $\wt_M(u,v) \le \alpha_u + \alpha_v$ for every $(u,v) \in E$. There cannot be any edge between 2 nodes left unmatched in
$M$ as that would contradict $M$'s popularity. So $\wt_M(u,v) \le  \alpha_u + \alpha_v \le -1$ for all 
$(u,v) \in E \cap (L \times L)$. Since $\wt_M(u,v) \in \{0,\pm 2\}$, $\wt_M(u,v) = -2$ for all edges 
$(u,v)$ in $L \times L$. In other words, every edge in $L \times L$ is labeled $(-,-)$.

Moreover, any $(+,+)$ edge can be present only
in $R \times R$ since $\wt_M(u,v) \le \alpha_u + \alpha_v \le 1$ for all edges $(u,v) \in L \times R$.
Finally, since $\alpha_u = \wt_M(u,u) = 0$ for all $u$ 
unmatched in $M$ (by complementary slackness conditions on the dual LP)
and every node $u \in R$ satisfies $\alpha_u = 1$, it means that all nodes in $R$ are matched in $M$. \qed
\end{proof}

\subsection{Correctness of our algorithm}

We will first show that if our algorithm returns a matching $M$, then $M$ is a strongly
dominant matching in $G$.

\begin{lemma}
  \label{lemma2:main}
If $M'$ is a stable matching in $G'$ then the projection of $M'$ is a strongly dominant matching in $G$.
\end{lemma}
\begin{proof}
Let $M$ be the projection of $M'$. In order to show that $M$ is a strongly dominant matching in $G$,
we will construct a witness $\vec{\alpha}$ as given in Theorem~\ref{lem:strongly-dominant}.

Set $\alpha_u = 0$ for all nodes $u$ left unmatched in $M$. For each node $u$ matched in $M$:
  \begin{itemize}
  \item  if $(u^+,\ast) \in M'$ then set $\alpha_u = 1$; else set $\alpha_u = -1$.
  \end{itemize}

  Note that $\sum_{u\in V} \alpha_u = 0$ since for each edge $(a,b) \in M$ (so either $(a^+,b^-)$ or $(a^-,b^+)$ is in $M'$),
  we have $\alpha_a + \alpha_b = 0$ and for each node $u$ that is unmatched in $M$, we have $\alpha_u = 0$ . We also have
  $\alpha_u \ge \wt_M(u,u)$ for all $u \in V$ since 
  (i)~$\alpha_u = 0 = \wt_M(u,u)$ for all $u$ left unmatched in $M$ and (ii)~$\alpha_u \ge -1 = \wt_M(u,u)$ for all $u$ matched in $M$.
  
We will now  show that for every $(a,b) \in E$, $\alpha_a + \alpha_b \ge \wt_M(a,b)$.

\begin{enumerate}
  \item Suppose $(a^+,\ast) \in M'$. So $\alpha_a = 1$. We will consider 3 subcases here. 
    \begin{itemize}
      \item The first subcase is that
      $(b^+,\ast) \in M'$. So $\alpha_b = 1$. Since $\wt_M(a,b) \le~2$, it follows that
      $\alpha_a + \alpha_b = 2 \ge \wt_M(a,b)$.
      \item The second subcase is that $(b^-,\ast) \in M'$. So $\alpha_b = -1$. If $(a^+,b^-) \in M'$ then  
           $\wt_M(a,b) = 0 = \alpha_a + \alpha_b$. So assume $(a^+,c^-)$ and $(b^-,d^+)$ belong to $M'$.       
       Since $M'$ is stable, the edge $(a^+,b^-)$ does not block $M'$. Thus either
       (i)~$a$ prefers $c^-$ to $b^-$ or (ii)~$b$ prefers $d^+$ to $a^+$. Hence $\wt_M(a,b) \in \{0,-2\}$ and so
       $\alpha_a + \alpha_b = 0 \ge \wt_M(a,b)$.
       \item The third subcase is that $b$ is unmatched in $M$. So $\alpha_b = 0$. 
       Since $M'$ is stable, the edge $(a^+,b^-)$ does not
       block $M'$. Thus $a$ prefers its partner in $M'$ to $b^-$ and so $\wt_M(a,b) = 0 < \alpha_a + \alpha_b$.
   \end{itemize}
   \item Suppose  $(a^-,\ast) \in M$. There are 3 subcases here as before. The case where $(b^+,\ast) \in M$ is
       totally analogous to the case where $(a^+,\ast)$ and $(b^-,\ast)$ are in $M$. So we will consider the remaining 2 subcases
       here.
    \begin{itemize}
       \item The first subcase is that $(b^-,\ast) \in M'$. So $\alpha_b = -1$. Let $(a^-,c^+)$ and $(b^-,d^+)$ belong to $M'$.
       Since $M'$ is stable, the edge $(a^+,b^-)$ does not block $M'$. So $b$ prefers $d^+$ to $a^+$.
       Similarly, the edge $(a^-,b^+)$ does not block $M'$. Hence $a$ prefers $c^+$ to $b^+$.
       Thus {\em both} $a$ and $b$ prefer their respective partners in $M$ to each other, i.e., $\wt_M(a,b) = -2$. So we have
       $\alpha_a + \alpha_b = -2 = \wt_M(a,b)$.

       \item The second subcase is that $b$ is unmatched in $M$. Then the edge $(a^+,b^-)$ {\em blocks} $M'$
       since $a$ prefers $b^-$ to $c^+$ (for any neighbor $c$) and $b$ prefers to be matched to $a^+$ than be left unmatched.
       Since $M'$ is stable and has no blocking edge, this means that this subcase does not arise.
    \end{itemize}   
 \item Suppose $a$ is unmatched in $M$. Then $(b^+,\ast) \in M'$ (otherwise $(a^-,b^+)$ blocks $M'$); moreover,
       $b$ prefers its partner in $M'$ to $a^-$. So we have $\wt_M(a,b) = 0 < \alpha_a + \alpha_b$.
     
\end{enumerate}
Thus we always have $\alpha_a + \alpha_b \ge \wt_M(a,b)$. 
Since $\vec{\alpha}$ satisfies the conditions in Theorem~\ref{lem:strongly-dominant}, $M$ is a strongly dominant matching in $G$.
\end{proof}

We will now show that if $G'$ has no stable matching, then $G$ has no strongly dominant matching.
\begin{lemma}
  \label{lemma1:main}
  If $G$ admits a strongly dominant matching then $G'$ admits a stable matching.
\end{lemma}
\begin{proof}
  Let $M$ be a strongly dominant matching in $G = (V,E)$. Let $\vec{\alpha}$ be a witness of $M$ as given in 
  Theorem~\ref{lem:strongly-dominant}. That is, $\alpha_u = 0$ for $u$ unmatched in $M$ and $\alpha_u = \pm 1$ for $u$ matched in $M$. As done in the proof of Theorem~\ref{lem:strongly-dominant}, we can interpret $(\tilde{M},\vec{\alpha})$ as a pair of optimal primal and dual solutions. Hence, for each $(u,v) \in M$, $\alpha_u + \alpha_v = \wt_M(u,v) = 0$ by complementary slackness on the dual LP,
  so $(\alpha_u,\alpha_v)$ is either $(1,-1)$ or $(-1,1)$.

  We will construct a stable matching $M'$ in $G'$ as follows. For each $(u,v) \in M$: 
  \begin{itemize}
  \item if $(\alpha_u,\alpha_v) = (1,-1)$ then add $(u^+,v^-)$ to $M'$;
        else add $(u^-,v^+)$ to $M'$.
  \end{itemize}

We will show that no edge in $E'\setminus M'$ blocks $M'$.
Let $(a^+,b^-) \notin M'$. We consider the following cases here:

\smallskip
  
\noindent{\bf Case 1.} Suppose $\alpha_b = 1$. Then $(b^+,d^-) \in M'$ where $d = M(b)$. Since $b$ prefers $d^-$ to $a^+$, 
                         $(a^+,b^-)$ is not a blocking edge to $M'$.
\smallskip
  
\noindent{\bf Case 2.} Suppose $\alpha_b = -1$. Then  $(b^-,d^+) \in M'$ where $d = M(b)$.
  We have 2 sub-cases here: (i)~$\alpha_a = 1$ and (ii)~$\alpha_a = -1$.
  Note that $\alpha_a \ne 0$ as the edge $(a,b)$ would not be covered by $\alpha_a + \alpha_b$ then. This is because if $\alpha_a = 0$ then
  $a$ is unmatched in $M$ and $\wt_M(a,b) = 0$ while $\alpha_a + \alpha_b = -1$.
  \begin{itemize}
  \item In sub-case~(i), some edge $(a^+,c^-)$ belongs to $M'$. 
  We know that $\wt_M(a,b) \le \alpha_a + \alpha_b = 0$, so either
  (1)~$a$ prefers $M(a) = c$ to $b$ or (2)~$b$ prefers $M(b) = d$ to $a$. Hence either (1)~$a$ prefers $c^-$ to $b^-$ or
  (2)~$b$ prefers $d^+$ to $a^+$. Thus $(a^+,b^-)$ is not a blocking edge to $M'$.
  \item In sub-case~(ii), some edge $(a^-,c^+)$ belongs to $M'$. We know that $\wt_M(a,b) \le \alpha_a + \alpha_b = -2$, so
  $a$ prefers $M(a) = c$ to $b$ {\em and} $b$ prefers $M(b) = d$ to $a$. Thus $b$ prefers $d^+$ to $a^+$, hence  
  $(a^+,b^-)$ is not a blocking edge to $M'$.
  \end{itemize}
  
\noindent{\bf Case 3.} Suppose $\alpha_b = 0$. Thus $b$ was unmatched in $M$. Each of $b$'s neighbors has to be matched in $M$ to a
  neighbor that it prefers to $b$, otherwise $M$ would be unpopular. We have $\alpha_a + \alpha_b \ge \wt_M(a,b) = 0$, hence it follows
  that $\alpha_a = 1$. Thus $(a^+,c^-) \in M'$ where $c$ is a neighbor that $a$ prefers to $b$. So $(a^+,b^-)$ is not a blocking edge
  to $M'$. \qed
\end{proof}

Lemmas~\ref{lemma2:main} and \ref{lemma1:main} show that a strongly dominant matching is present in $G$ if and only if a stable
matching is present in $G'$. This finishes the proof of correctness of our algorithm.
Since Irving's stable matching algorithm in $G'$ can be implemented to run in linear time~\cite{Irv85}, we can conclude the following theorem.

\begin{theorem}
  \label{thm:strongly-dom}
  There is a linear time algorithm to determine if a graph $G = (V,E)$ with strict preference lists admits
  a strongly dominant matching or not and if so, to return one.
\end{theorem}

\section{Finding a popular matching that is neither stable nor dominant}
\label{sec:hardness}

This section is devoted to proving the following result.

\begin{theorem}
	\label{final-thm}
	Given a bipartite instance $G = (A \cup B,E)$ with strict preference lists, the problem of deciding if $G$ admits a popular matching that is neither
	a stable matching nor a dominant matching is NP-hard.
\end{theorem}

Our reduction will be from 1-in-3 SAT. Recall that 1-in-3 SAT is the set of 3CNF formulas with no negated variables such that there is a
satisfying assignment that makes {\em exactly one} variable true in each clause. Given an input formula $\phi$, to determine if $\phi$ is
1-in-3 satisfiable or not is NP-hard~\cite{Sch78}.

We will now build a bipartite instance $G = (A \cup B, E)$. 
The node set $A \cup B$ will consist of nodes in 4 levels: levels~0, 1, 2, and 3 along with
4 nodes $a_0,b_0,z$, and $z'$.
Nodes in $(A\cup B) \setminus \{a_0,b_0,z,z'\}$ are partitioned into gadgets that appear in some level~$i$, for $i \in \{0,1,2,3\}$. For each variable in our 1-in-3 SAT formula, we construct a variable gadget (in level 1), and for each clause, we construct three clause gadgets in level 0, three in  level~2, and one in level~3.

We will show that every gadget forms a separate connected component in the popular subgraph of $G$. 
If $M$ is any popular matching in $G$ that is neither a stable matching nor a dominant matching and $\vec{\alpha}$ is any witness of $M$ then it will be the case that every level~0 gadget in $G$ is in zero state in $\vec{\alpha}$ and every level~3 gadget in $G$ is in unit state in $\vec{\alpha}$. This will force the following property to hold for every clause in $\phi$:
\begin{itemize}
\item if $c = X_i \vee X_j \vee X_k$ then among the gadgets corresponding to $X_i,X_j,X_k$ in level~1, {\em exactly one} is in unit state in $\vec{\alpha}$.
\end{itemize}

Thus a popular matching in $G$ that is neither stable nor dominant will yield a 1-in-3 satisfiable assignment to $\phi$. Conversely, if $\phi$ is 1-in-3 satisfiable then we can build a popular matching in $G$ that is neither stable nor dominant. We describe our gadgets below.

\medskip

\noindent{\em Level~1 nodes.} Every gadget in level~1 is a variable gadget.
Corresponding to each variable $X_i$, we will have the gadget in Fig.~\ref{level1:example}. 
The preference lists of the 4 nodes in the gadget corresponding to $X_i$ are as follows:

\begin{minipage}[c]{0.45\textwidth}
			
			\centering
			\begin{align*}
			        &x_i\colon \, y_i \succ y'_i \succ  z \succ \cdots   \qquad\qquad &&  y_i\colon \, x_i \succ x'_i \succ z' \succ \cdots \\
                                &x'_i\colon \, y_i \succ y'_i \succ \cdots   \qquad\qquad &&  y'_i\colon \, x_i \succ x'_i \succ \cdots\\
			\end{align*}
\end{minipage}

\smallskip

The nodes in the gadget corresponding to $X_i$ are also adjacent to nodes in the clause gadgets:
these neighbors belong to the ``$\cdots$'' part of the preference lists. Note that the order among the nodes in the ``$\cdots$'' part
in the above preference lists does not matter. 

\begin{figure}[h]
\centerline{\resizebox{0.18\textwidth}{!}{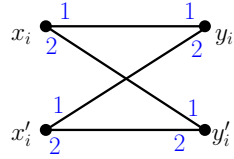}}
\caption{The gadget corresponding to variable $X_i$: node preferences are indicated on edges. The node $y_i$ is the top choice of both $x_i$ and $x'_i$ and the node $y'_i$ is the second choice of both $x_i$ and $x'_i$. The node $x_i$ is the top choice of both $y_i$ and $y'_i$ and the node $x'_i$ is the second choice of both $y_i$ and $y'_i$.}
\label{level1:example}
\end{figure}

Let $c = X_i \vee X_j \vee X_k$ be a clause in $\phi$. We will describe the gadgets that correspond to $c$.
For the sake of readability, when we describe preference lists below, we drop the superscript $c$ from
all the nodes appearing in gadgets corresponding to clause $c$.

\medskip

\noindent{\em Level~0 nodes.} 
There will be three level~0 gadgets, each on 4 nodes, corresponding to clause $c$. See Fig.~\ref{level0:example}.
We describe below the preference lists of the 4 nodes $a^c_1,b^c_1,a^c_2,b^c_2$ that belong to the leftmost gadget.

\begin{minipage}[c]{0.45\textwidth}
			
			\centering
			\begin{align*}
			        &a_1\colon \, b_1 \succ \underline{y'_j} \succ b_2 \succ \underline{z}  \qquad\qquad && b_1\colon \, a_2 \succ \underline{x'_k} \succ a_1 \succ \underline{z'} \\
                                &a_2\colon \, b_2 \succ b_1  \qquad\qquad && b_2\colon \, a_1 \succ a_2 \\
			\end{align*}
\end{minipage}

Neighbors that are outside this gadget are underlined.
The preferences of nodes in the other two gadgets in level~0 corresponding to $c$ ($a^c_t,b^c_t$ for $t = 3,4$ and $a^c_t,b^c_t$ for $t = 5,6$)
are analogous. 

\begin{figure}[h]
\centerline{\resizebox{0.74\textwidth}{!}{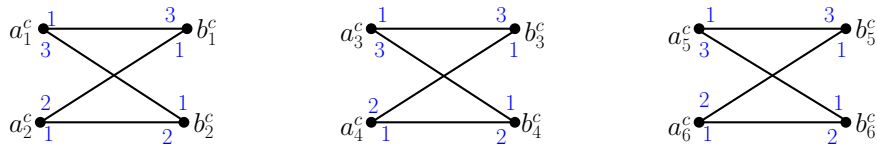}}
\caption{Corresponding to clause $c =  X_i \vee X_j \vee X_k$, we have the above 3 gadgets in level 0. The node $a^c_1$'s second choice
  is $y'_j$ and $b^c_1$'s is $x'_k$, similarly, $a^c_3$'s is $y'_k$ and $b^c_3$'s is $x'_i$, also $a^c_5$'s is $y'_i$ and $b^c_5$'s is $x'_j$.}
\label{level0:example}
\end{figure}

We will now describe the three level~2 gadgets corresponding to clause $c$. See Fig.~\ref{level2:example}.

\medskip

\noindent{\em Level~2 nodes.}
There will be three level~2 gadgets, each on 6 nodes, corresponding to clause $c$.
The preference lists of the nodes $p^c_t,q^c_t$ for $0 \le t \le 2$ are described below.

\begin{minipage}[c]{0.45\textwidth}
			
			\centering
			\begin{align*}
			        &p_0\colon \, q_0 \succ q_2  \qquad\qquad && q_0\colon \, p_0 \succ p_2 \succ \underline{z'} \succ \underline{s_0}\\
                                &p_1\colon \, q_1 \succ q_2 \succ \underline{z}  \qquad\qquad && q_1\colon \, p_1 \succ p_2 \\
                                &p_2\colon \, q_0 \succ \underline{y_j} \succ q_1 \succ q_2  \succ \cdots\qquad\qquad && q_2\colon \, p_1 \succ \underline{x_k} \succ p_0 \succ p_2 \succ \cdots\\
			\end{align*}
\end{minipage}

The ``$\cdots$'' in the preference lists of $p_2$ and $q_2$ above are to nodes $t^{c_i}_0$ and $s^{c_i}_0$ respectively (in a level 3 gadget), 
for {\em all} clauses $c_i$. The order among these neighbors is not important.

\begin{figure}[h]
\centerline{\resizebox{0.8\textwidth}{!}{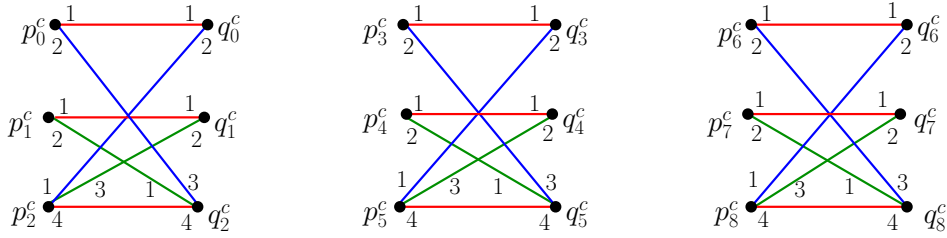}}
\caption{We have the above 3 gadgets in level 2 corresponding to $c =  X_i \vee X_j \vee X_k$ . The node $p^c_2$'s second choice is
   $y_j$ and $q^c_2$'s is $x_k$, similarly, $p^c_5$'s is $y_k$ and $q^c_5$'s is $x_i$, similarly $p^c_8$'s is $y_i$ and $q^c_8$'s is $x_j$.}
\label{level2:example}
\end{figure}

Let us note the preference lists of $p_2$ and $q_2$: they are each other's fourth choices.
The node $p_2$ regards $q_0$ as its top choice, $y_j$ as its second choice, and $q_1$ as its third choice.
The node $q_2$ regards $p_1$ as its top choice, $x_k$ as its second choice, and $p_0$ as its third choice.

The preferences of nodes $p^c_t,q^c_t$ for $3 \le t \le 5$
are described below. The ``$\cdots$'' in the preference lists of $p_5$ and $q_5$ above are to nodes $t^{c_i}_0$ and 
$s^{c_i}_0$ respectively, for all clauses $c_i$.

\begin{minipage}[c]{0.45\textwidth}
			
			\centering
			\begin{align*}
			        &p_3\colon \, q_3 \succ q_5  \qquad\qquad && q_3\colon \, p_3 \succ p_5 \succ \underline{z'} \succ \underline{s_0}\\
                                &p_4\colon \, q_4 \succ q_5 \succ \underline{z} \succ \underline{t_0} \qquad\qquad && q_4\colon \, p_4 \succ p_5 \\
                                &p_5\colon \, q_3 \succ \underline{y_k} \succ q_4 \succ q_5  \succ \cdots\qquad\qquad && q_5\colon \, p_4 \succ \underline{x_i} \succ p_3 \succ p_5 \succ \cdots\\
			\end{align*}
\end{minipage}

The preferences of nodes $p^c_t,q^c_t$ for $6 \le t \le 8$
are described below. The ``$\cdots$'' in the preference lists of $p_8$ and $q_8$ above are to nodes $t^{c_i}_0$ and $s^{c_i}_0$ 
respectively, for all clauses $c_i$. 

\begin{minipage}[c]{0.45\textwidth}
			
			\centering
			\begin{align*}
			        &p_6\colon \, q_6 \succ q_8  \qquad\qquad && q_6\colon \, p_6 \succ p_8 \succ \underline{z'} \\
                                &p_7\colon \, q_7 \succ q_8 \succ \underline{z} \succ \underline{t_0} \qquad\qquad && q_7\colon \, p_7 \succ p_8 \\
                                &p_8\colon \, q_6 \succ \underline{y_i} \succ q_7 \succ q_8  \succ \cdots\qquad\qquad && q_8\colon \, p_7 \succ \underline{x_j} \succ p_6 \succ p_8 \succ \cdots\\
			\end{align*}
\end{minipage}

\noindent{\em Level~3 nodes.}
Gadgets in level~3 are again clause gadgets. There is exactly one level~3 gadget on 8 nodes $s^c_i,t^c_i$, for $0 \le i \le 3$,
corresponding to clause $c$. 

\begin{minipage}[c]{0.45\textwidth}
			
			\centering
			\begin{align*}
				&s_0\colon \, t_1  \succ \underline{q_0} \succ t_2\succ \underline{q_3} \succ t_3 \succ \cdots \qquad\qquad && t_0\colon \, s_3  \succ \underline{p_7} \succ s_2\succ  \underline{p_4} \succ s_1 \succ \cdots\\
			        &s_1\colon \, t_1 \succ t_0  \qquad\qquad && t_1\colon \, s_1 \succ s_0 \\
                                &s_2\colon \, t_2 \succ t_0  \qquad\qquad && t_2\colon \, s_2 \succ s_0 \\
                                &s_3\colon \, t_3 \succ t_0  \qquad\qquad && t_3\colon \, s_3 \succ s_0 \\
			\end{align*}
\end{minipage}

The preference lists of the 8 nodes in the level~3 gadget corresponding to clause $c$ are described above.
It is important to note the preference lists of $s_0$ and $t_0$ here.

Among neighbors in this gadget, $s_0$'s order is $t_1 \succ t_2 \succ t_3$ while
$t_0$'s order is $s_3 \succ s_2 \succ s_1$. Also, $s_0$'s order is interleaved with $q_0 \succ q_3$ (these are nodes from level~2 gadgets) and
$t_0$'s order is interleaved with $p_7 \succ p_4$.

The ``$\cdots$'' in the preference lists of $s_0$ and $t_0$ above are to neighbors in levels~1 and 2. Let $n_0$ be the number of variables in $\phi$.
All the nodes $y'_1,\ldots,y'_{n_0}$ along with $q^{c_i}_2,q^{c_i}_5,q^{c_i}_8$ for all clauses $c_i$ will be at the tail of the preference list
of $s^c_0$ and the order among all these nodes is not important.
Similarly, all the nodes $x'_1,\ldots,x'_{n_0}$ along with $p^{c_i}_2,p^{c_i}_5,p^{c_i}_8$ for all clauses $c_i$
will be at the tail of the preference list of $t^c_0$ and the order among all these nodes is also not important. 

\medskip

There are four more nodes in $G$. These are $a_0,z' \in A$ and $b_0,z \in B$.
Thus we have
\begin{eqnarray*}
A & = & \cup_c\{a^c_i: 1 \le i \le 6\} \cup_i \{x_i,x'_i\} \cup_c\{p^c_i: 0 \le i \le 8\} \cup_c\{s^c_i: 0 \le i \le 3\} \cup \{a_0,z'\}\\
B & = & \cup_c\{b^c_i: 1 \le i \le 6\} \cup_i \{y_i,y'_i\} \cup_c\{q^c_i: 0 \le i \le 8\} \cup_c\{t^c_i: 0 \le i \le 3\} \cup \{b_0,z\}.
\end{eqnarray*}

The neighbors of $a_0$ are $b_0,z$ and the neighbors of $b_0$ are $a_0,z'$.
The node $a_0$'s preference list is $b_0 \succ z$ and the node $b_0$'s preference list is $a_0 \succ z'$.

\smallskip

The set of neighbors of $z$ is $\{a_0\} \cup_i \{x_i\} \cup_c \{a^c_1,a^c_3,a^c_5\} \cup_c \{p^c_1,p^c_4,p^c_7\}$ and the set of neighbors of $z'$ is
$\{b_0\} \cup_i \{y_i\} \cup_c \{b^c_1,b^c_3,b^c_5\} \cup_c \{q^c_0,q^c_3,q^c_6\}$.
The preference lists of $z$ and $z'$ are as follows: (here $k$ is the number of clauses in $\phi$)
\begin{eqnarray*}
  z &\colon&  x_1 \succ \cdots \succ  x_{n_0} \ \succ \ p^{c_1}_1 \succ \cdots \succ p^{c_k}_7 \ \succ \ a_0 \succ \cdots \\
  z' &\colon& y_1 \succ \cdots \succ  y_{n_0} \ \succ \ q^{c_1}_0 \succ \cdots \succ q^{c_k}_6 \ \succ \ b_0 \succ \cdots
\end{eqnarray*}  

Thus $z$ prefers neighbors in level~1 to neighbors in level~2, then comes $a_0$, and then neighbors in level~0.
Analogously, for $z'$ (with $b_0$ replacing $a_0$). 
The order among neighbors in level~$i$ (for $i = 0,1,2$) in the preference lists of $z$ and $z'$ does not matter.

\subsection{Some stable/dominant matchings in $G$}

It would be helpful to see  some stable matchings and dominant matchings in the above instance $G$.
\begin{itemize}
\item The men-optimal stable matching $S$ in $G$ includes $(a_0,b_0)$ and in the level~0 gadgets, for all clauses $c$,
the edges $(a^c_i,b^c_i)$ for $1 \le i \le 6$.
  \begin{itemize}
   \item  In the level~1 gadgets, the edges $(x_i,y_i)$ and $(x'_i,y'_i)$ are included for all $i \in [n_0]$.

   \item In the level~2 gadgets, for all clauses $c$, the edges $(p^c_i,q^c_i)$ for $0 \le i \le 8$ are included.

   \item In the level~3 gadgets, for all clauses $c$, the edges $(s^c_i,t^c_i)$ for $1 \le i \le 3$ are included.

   \item The nodes $z,z'$ and $s^c_0,t^c_0$ for all clauses $c$ are left unmatched in $S$.
  \end{itemize}

  \smallskip
  
\item The women-optimal stable matching $S'$ in $G$ includes $(a_0,b_0)$ and the same edges as $S$ in all level~1, 2, 3 gadgets. In 
level~0, $S'$ includes for all clauses $c$, the edges $(a^c_1,b^c_2),(a^c_2,b^c_1)$,$(a^c_3,b^c_4)$,$(a^c_4,b^c_3)$, $(a^c_5,b^c_6),(a^c_6,b^c_5)$.

  \smallskip

\item The dominant matching $M^*$ as computed by the algorithm in \cite{Kav12} will be as follows:
  \begin{itemize}
  \item $M^*$ contains the edges $(a_0,z), (z',b_0)$, and in the level~0 gadgets, for all clauses $c$,
the edges $(a^c_1,b^c_2),(a^c_2,b^c_1)$,$(a^c_3,b^c_4)$,$(a^c_4,b^c_3)$, $(a^c_5,b^c_6),(a^c_6,b^c_5)$.
  \item   In the level~1 gadgets, the edges $(x_i,y'_i)$ and $(x_i,y'_i)$ are included for all $i \in [n_0]$.
  \item   In the level~2 gadgets, for each clause $c$, the edges $(p^c_0,q^c_2),(p^c_1,q^c_1),(p^c_2,q^c_0)$ are included from the leftmost gadget
    (see Fig.~\ref{level2:example}). Analogous edges (two blue ones and the middle red edge) are included from the other two level~2 gadgets corresponding to $c$.
  \item In the level~3 gadgets, the edges $(s^c_0,t^c_1), (s^c_1,t^c_0), (s^c_2,t^c_2), (s^c_3,t^c_3)$ for all clauses $c$ are included.
  \end{itemize}
\end{itemize}

Note that $M^*$ is a perfect matching as it matches all nodes in $G$.  We can show the above matching $M^*$ to be popular by the following witness
$\vec{\alpha} \in \{\pm 1\}^n$ to $M^*$:
\begin{itemize}
\item $\alpha_{a_0} = \alpha_{b_0} = 1$ while $\alpha_{z} = \alpha_{z'} = -1$.
\item $\alpha_{a^c_i} = 1$ and $\alpha_{b^c_i} = -1$ for $1 \le i \le 6$ and all clauses $c$.
\item $\alpha_{x_i} = \alpha_{y_i} = 1$ while $\alpha_{x'_i} = \alpha_{y'_i} = -1$ for all $i \in [n_0]$.
\item $\alpha_{p^c_0} = \alpha_{q^c_0} = \alpha_{p^c_1} = 1$ while $\alpha_{q^c_1} = \alpha_{p^c_2} = \alpha_{q^c_2} = -1$ for all clauses $c$. Similarly
  for the other 2 level~2 gadgets corresponding to $c$ and all other clauses.
\item $\alpha_{s^c_1} = \alpha_{t^c_1} = \alpha_{s^c_2} = \alpha_{s^c_3} = 1$ while $\alpha_{s^c_0} = \alpha_{t^c_0} = \alpha_{t^c_2} = \alpha_{t^c_3} = -1$ for all clauses $c$.
\end{itemize}

It can be checked that we have $\alpha_u + \alpha_v = 0$ for every edge $(u,v) \in M^*$. We also have $\alpha_u + \alpha_v \ge \wt_{M^*}(u,v)$ for every edge $(u,v)$
in the graph. In particular, the endpoints of every blocking edge to $M^*$, such as $(a_0,b_0)$, $(x_i,y_i)$ for all $i$,
$(p^c_{3j}, q^c_{3j})$ for all $c$ and $j \in \{0,1,2\}$, and $(s^c_1, t^c_1)$ for all $c$, have their $\alpha$-value equal to 1.

There are many other dominant matchings in this instance $G$:
\begin{itemize}
\item The edges $(a^c_i,b^c_i)$ may be included for $i \in \{1,\cdots,6\}$ and all clauses $c$.
 
\item The top red edge and two green ones (such as the edges $(p^c_0,q^c_0),(p^c_1,q^c_2),(p^c_2,q^c_1)$
  in the leftmost level~2 gadget corresponding to $c$) may be included from a level~2 gadget.
\item  From the level~3 gadget corresponding to $c$,
the edges $(s^c_0,t^c_2), (s^c_1,t^c_1), (s^c_2,t^c_0), (s^c_3,t^c_3)$ or the edges $(s^c_0,t^c_3), (s^c_1,t^c_1), (s^c_2,t^c_2), (s^c_3,t^c_0)$ may be included.
\end{itemize}

\subsection{The popular subgraph of $G$}

Recall the popular subgraph $F_G$ from Section~\ref{prelims}, whose edge set is the set of popular edges in $G$.

\begin{lemma}
  \label{lem:conn-comp}
   Let $C$ be any level~$i$ gadget in $G$, where $i \in \{0,1,2,3\}$. All the nodes in $C$ belong to the same connected component in $F_G$.
\end{lemma}
\begin{proof}
  Consider a level~0 gadget in $G$, say on $a^c_1,b^c_1,a^c_2,b^c_2$ (see Fig.~\ref{level0:example}). The men-optimal stable matching $S$ in $G$ contains
  the edges $(a^c_1,b^c_1)$ and $(a^c_2,b^c_2)$ while the women-optimal stable matching $S'$ contains the edges $(a^c_1,b^c_2)$ and $(a^c_2,b^c_1)$.
  Thus there are popular edges among these 4 nodes and so these 4 nodes belong to the same connected component in $F_G$.

  Consider a level~1 gadget in $G$, say on $x_i,y_i,x'_i,y'_i$ (see Fig.~\ref{level1:example}). Every stable matching in $G$ contains $(x_i,y_i)$ and $(x'_i,y'_i)$
  while the dominant matching $M^*$ contains $(x_i,y'_i)$ and $(x'_i,y_i)$. Thus there are popular edges among these 4 nodes
  and so these 4 nodes belong to the same connected component in $F_G$.

  Consider a level~2 gadget in $G$, say on $p^c_i,q^c_i$ for $i = 0,1,2$ (see Fig.~\ref{level2:example}).
  The dominant matching $M^*$ contains the edges $(p^c_0,q^c_2)$ and $(p^c_2,q^c_0)$.
  There is also another dominant matching in $G$ that contains the edges $(p^c_1,q^c_2)$ and $(p^c_2,q^c_1)$.
  Thus there are popular edges among these 6 nodes and so these 6 nodes  belong to the same connected component in $F_G$.

  Consider a level~3 gadget in $G$, say on $s^c_i,t^c_i$ for $i = 0,1,2,3$.
  The dominant matching $M^*$ contains $(s^c_0,t^c_1)$, and $(s^c_1,t^c_0)$.
  There is another dominant matching in $G$ that contains $(s^c_0,t^c_2)$ and $(s^c_2,t^c_0)$. 
  There is yet another dominant matching in $G$ that contains $(s^c_0,t^c_3)$ and $(s^c_3,t^c_0)$. 
  Thus there are popular edges among these 8 nodes and so these 8 nodes belong to the same connected component in $F_G$. \qed
\end{proof}  

The following theorem will be important for us and we will prove it in Section~\ref{sec:proof-separate}.

\begin{theorem}
  \label{thm:separate}
  Every level~$i$ gadget, for $i \in \{0,1,2,3\}$, forms a distinct connected component in the graph $F_G$. The four nodes $a_0,b_0,z$, and $z'$
  belong to their own connected component in $F_G$.
\end{theorem}

\subsection{Popular matchings in $G$}

 Let $M$ be any popular matching in $G$. Note that $M$ either matches {\em both} $z$ and $z'$ or leaves both these nodes unmatched.
 This is because both $z$ and $z'$ are unstable nodes in the same connected component in $F_G$ (by Theorem~\ref{thm:separate}), 
 so either both are matched or both are unmatched in $M$ (by Lemma~\ref{prop1}). It is similar with nodes $s^c_0$ and $t^c_0$ 
 for any clause $c$: any popular matching either matches both $s^c_0$ and $t^c_0$ or leaves both these nodes unmatched.
 
\begin{lemma}
  \label{bipartite:lemma2}
  Suppose $M$ is a popular matching in $G$ that matches $z$ and $z'$. Then $M$ is a dominant matching in $G$.
\end{lemma}
\begin{proof}
  Let $\vec{\alpha} \in \{0,\pm 1\}^n$ be a witness of $M$. It follows from
  Theorem~\ref{thm:separate} that $(a_0,z)$ and $(z',b_0)$ are in $M$.
  Since $z$ and $z'$ prefer their neighbors in level~1 
  to $a_0$ and $b_0$ respectively while these neighbors prefer their partners in $M$ to $z$ and $z'$ (by Theorem~\ref{thm:separate}),
  we have $\wt_M(x_i,z) = \wt_M(z',y_i) = 0$. 

  The nodes $z$ and $z'$ are unstable in $G$ and $M$ matches them, so $\alpha_z = \alpha_{z'} = -1$ (by Lemma~\ref{prop0}).
  Since $\alpha_{x_i} + \alpha_z \ge 0$ and  $\alpha_{z'} + \alpha_{y_i} \ge 0$, it follows that $\alpha_{x_i} = \alpha_{y_i} = 1$ for all $i$.  
  Thus all nodes in level~1 have $\alpha$-values equal to $\pm 1$ (by Lemma~\ref{prop1}).
  In particular, $\alpha_{x'_i} = \alpha_{y'_i} = -1$. This is due to the fact that $\wt_M(x_i,y'_i) = 0$ (by Theorem~\ref{thm:separate})
  and $\alpha_{x_i} + \alpha_{y'_i} = \wt_M(x_i,y'_i)$ (by Lemma~\ref{prop0}) as $(x_i,y'_i)$ is a popular edge. 
  Similarly, with $(x'_i,y_i)$.

  Suppose $s_0^c,t_0^c$ are unmatched in $M$. Then $\wt_M(s^c_0,y'_i) = \wt_M(x'_i,t^c_0) = 0$. This is because
  $s^c_0$ and $t^c_0$ prefer to be matched to any neighbor than be unmatched while (by Theorem~\ref{thm:separate}) 
  $y'_i$ and $x'_i$ prefer their partners in $M$ to $s^c_0$ and $t^c_0$, respectively. Since we assumed $s_0^c,t_0^c$ to be unmatched in $M$,
  $\alpha_{s_0^c} = \alpha_{t^c_0} = 0$ (by Lemma~\ref{prop0}). So this implies that $\alpha_{s_0} + \alpha_{y'_i} = 0 -1 < \wt_M(s^c_0,y'_i)$, i.e.,
  the edge $(s^c_0,y'_i)$ (similarly, $(x'_i,t^c_0)$) is not covered by the sum of $\alpha$-values of its endpoints, a contradiction.
  Thus $s_0^c,t_0^c$ are forced to be matched in $M$.
  
  Thus $s^c_0$ and $t^c_0$ for all clauses $c$
  are matched in $M$. The other nodes in $G$ are stable and hence they have to be matched in $M$.
  Thus $M$ is a perfect matching and also popular, so it is a dominant matching in $G$. \qed
\end{proof}

\begin{lemma}
  \label{bipartite:lemma3}
  Suppose $M$ is a popular matching in $G$ that leaves $s^c_0$ and $t^c_0$ unmatched for some $c$. Then $M$ is a stable matching 
  in $G$.
\end{lemma}
\begin{proof}
  We will repeatedly use Theorem~\ref{thm:separate} here.
  Let $\vec{\alpha} \in \{0,\pm 1\}^n$ be a witness of $M$. Since the nodes $s^{c}_0$ and $t^{c}_0$ are unmatched in $M$, we have
  $\alpha_{s^c_0} = \alpha_{t^c_0} = 0$ (by Lemma~\ref{prop0}). Also $\wt_M(s^c_0,q^c_2) = \wt_M(p^c_2,t^c_0) = 0$ since both $s^c_0$ and 
  $t^c_0$ prefer to be matched than be unmatched while $q^c_2$ and $p^c_2$ prefer their partners in $M$ to $s^c_0$ and $t^c_0$, 
  respectively (by Theorem~\ref{thm:separate}). So $\alpha_{p^c_2} \ge 0$ and $\alpha_{q^c_2} \ge 0$ for all $c$.

  We also have $\alpha_{p^c_2} + \alpha_{q^c_2} = \wt_M(p^c_2, q^c_2) \le 0$ since any popular matching
  matches $p^c_2$ to a partner at least as good as $q^c_2$  and similarly, $q^c_2$ to a partner at least as good as $p^c_2$
  (by Theorem~\ref{thm:separate}).
  This means that $\alpha_{p^c_2} = \alpha_{q^c_2} = 0$. The same argument can be used for every $p^{c_i}_{3j+2}$ and $q^{c_i}_{3j+2}$
  (for any clause $c_i$ and $j = 0,1,2$) to show that $\alpha_{p^{c_i}_{3j+2}} = \alpha_{q^{c_i}_{3j+2}} = 0$. Thus all level~2 nodes
  have $\alpha$-values equal to 0 (by Lemmas~\ref{prop1} and \ref{lem:conn-comp}).

  The fact that all level~2 nodes have $\alpha$-values equal to 0 immediately implies that all level~3 nodes also have 
  $\alpha$-values equal to 0. This is because if $M$ matches $s^{c_i}_0$ and $t^{c_i}_0$ for some clause $c_i$ then at least one of 
  $s^{c_i}_0,t^{c_i}_0$ is not matched to its top choice
  neighbor in its gadget. So either $\wt_M(s^{c_i}_0, q^{c_i}_0) = 0$ or $\wt_M(p^{c_i}_7,t^{c_i}_0) = 0$. 
  Since $\alpha_{q^{c_i}_0} = \alpha_{p^{c_i}_7} = 0$ (these are level~2 nodes) and $\alpha_{s^{c_i}_0} = \alpha_{t^{c_i}_0} = -1$ 
  (by Lemma~\ref{prop0}), we have a
  contradiction. Thus $\alpha_u = 0$ for every level~3 node $u$ (by Lemma~\ref{prop1}). 

  Similarly, $\alpha_{x'_i} \ge 0$ and $\alpha_{y'_i} \ge 0$ for all $r$ as the edges $(x'_i,t^c_0)$ and
  $(s_0^c,y'_i)$ would not be covered otherwise. Also, $\alpha_{x'_i} + \alpha_{y'_i} = \wt_M(x'_i, y'_i)$ (by Lemma~\ref{prop0}) as
  $(x'_i,y'_i)$ is a popular edge and $\wt_M(x'_i, y'_i) \le 0$ since $M$ matches $x'_i$ to either $y_i$ or $y'_i$
  and similarly, $y'_i$ to either $x_i$ or $x'_i$ (by Theorem~\ref{thm:separate}). Thus $\alpha_{x'_i} = \alpha_{y'_i} = 0$.
  This means that all level~1 nodes have $\alpha$-values equal to 0 (by Lemma~\ref{prop1}).
  
  Since all level~1 nodes have $\alpha$-values equal to 0, we have $\alpha_z = \alpha_{z'} = 0$; otherwise
  the  edges $(x_i,z)$ and $(z',y_i)$ would not be covered. This is because $\wt_M(x_i,z) = \wt_M(z',y_i) = 0$ 
(by Theorem~\ref{thm:separate}). So in order to cover the edges $(x_i,z)$ and $(z',y_i)$, we need to have $\alpha_z \ge 0$ 
and $\alpha_{z'} \ge 0$, i.e., $\alpha_z = \alpha_{z'} = 0$ (by Lemma~\ref{prop0}). 
Thus $\alpha_{a_0} = \alpha_{b_0} = 0$ (by Theorem~\ref{thm:separate} and Lemma~\ref{prop1}).

Moreover, $\alpha_z = \alpha_{z'} = 0$ also implies that all their neighbors in level~0 have their $\alpha$-values at least 0. For instance,
consider $a^c_1$ and $b^c_1$: in order to cover the edges $(a^c_1,z)$ and $(z',b^c_1)$, we have $\alpha_{a^c_1} \ge 0$ and $\alpha_{b^c_1} \ge 0$.
Since either $(a^c_1,b^c_1) \in M$ or $(a^c_1,b^c_2),(a^c_2,b^c_1)$ are in $M$ (by Theorem~\ref{thm:separate}), we have
$\wt_M(a^c_1,b^c_1) = 0$. Because $(a^c_1,b^c_1)$ is a popular edge, this means $\alpha_{a^c_1} + \alpha_{b^c_1} = 0$ (by Lemma~\ref{prop0}).
Thus $\alpha_{a^c_1} = \alpha_{b^c_1} = 0$.

Similarly, $\alpha_{a^c_{2i-1}} = \alpha_{b^c_{2i-1}} = 0$ for $i = 1,2,3$ and all clauses $c$.
Thus all level~0 nodes have $\alpha$-values equal to 0.
So $\vec{\alpha} = \vec{0}$, i.e., $\wt_M(e) \le 0$ for all edges $e$. In other words, there is no blocking edge to $M$. 
Thus $M$ is a stable matching. \qed
\end{proof}  

It follows from Lemmas~\ref{bipartite:lemma2} and \ref{bipartite:lemma3} that if $M$ is a popular matching in $G$ that is
neither stable nor dominant then $M$ has to match nodes $s^c_0,t^c_0$ for all $c$ and leave $z$ and $z'$ unmatched.
Equivalently, $M$ has to match all nodes except $z$ and $z'$.

Conversely, if $M$ is a popular matching in $G$ that matches all nodes except $z$ and $z'$ then $M$ is neither a max-size popular matching
nor a min-size popular matching, i.e., $M$ is neither dominant nor stable. Thus we can conclude the following theorem.

\begin{theorem}
  \label{bipartite:theorem}
  The graph $G$ admits a popular matching that is neither stable nor dominant if and only if $G$ admits a popular matching
  that matches all nodes except $z$ and $z'$.
\end{theorem}

\subsection{Desired popular matchings in $G$}
We will call a matching $M$ in $G$ that matches all nodes except $z$ and $z'$ a {\em desired popular matching}
here. Let $M$ be such a matching and let $\vec{\alpha} \in\{0,\pm 1\}^n$ be a witness of $M$, where $n$ is the number of
nodes in $G$.

Recall Definition~\ref{def:stab-domn} from Section~\ref{prelims}.
We say a gadget is in {\em unit} (similarly, {\em zero}) state in $\vec{\alpha}$ if for any node $u$ in this gadget,
we have $\alpha_u \in \{\pm 1\}$ (resp., $\alpha_u = 0$). The following two observations will be important here. 

\begin{itemize}
\item[1.] All level~3 gadgets have to be in {\em unit} state in $\vec{\alpha}$.
\item[2.] All level~0 gadgets have to be in {\em zero} state in $\vec{\alpha}$.
\end{itemize}

  The nodes $s^c_0$ and $t^c_0$, for all clauses $c$, are left unmatched in any stable matching in $G$.
  Since $M$ has to match the unstable nodes $s^c_0$ and $t^c_0$ for all clauses $c$, 
  $\alpha_{s^c_0} = \alpha_{t^c_0} = -1$ for all $c$ (by Lemma~\ref{prop0}). Thus the first observation follows 
  from Lemmas~\ref{prop1} and \ref{lem:conn-comp}. We prove the second observation below.

\begin{new-claim}
  Any level~0 gadget has to be in zero state in $\vec{\alpha}$.
\end{new-claim}
\begin{proof}
  Consider any level~0 gadget, say on nodes $a^c_1,b^c_1,a^c_2,b^c_2$. Since $M$ is a popular matching, we have $\alpha_{a^c_1} + \alpha_z \ge \wt_M(a^c_1,z)$ and
  $\alpha_{z'} + \alpha_{b^c_1}  \ge \wt_M(z',b^c_1)$. Since $z$ and $z'$ are unmatched in $M$, $\alpha_{z} = \alpha_{z'} = 0$ (by Lemma~\ref{prop0}). 
  Since $\wt_M(a^c_1,z) = 0$ and $\wt_M(z',b^c_1) = 0$, we have $\alpha_{a^c_1} \ge 0$ and $\alpha_{b^c_1} \ge 0$.

  The edge $(a^c_1,b^c_1)$ is a popular edge. Thus $\alpha_{a^c_1} + \alpha_{b^c_1} = \wt_M(a^c_1,b^c_1)$ (by Lemma~\ref{prop0}) and we have 
  $\wt_M(a^c_1,b^c_1) = 0$ (as seen in the last part of the proof of Lemma~\ref{bipartite:lemma3}).
  Thus $\alpha_{a^c_1} + \alpha_{b^c_1} = 0$. So $\alpha_{a^c_1} = \alpha_{b^c_1} = 0$.
  Thus this gadget is in zero state and this holds for every level~0 gadget. \hfill $\lozenge$
\end{proof}

Lemmas~\ref{lemma1}-\ref{lemma3} are easy to show and are crucial to our NP-hardness proof. Let $c = X_i \vee X_j \vee X_k$ be any clause 
in $\phi$. In our proofs below, we are omitting the superscript $c$ from node names for the sake of readability. Recall that 
$\vec{\alpha} \in \{0, \pm 1\}^n$ is a witness of our ``desired popular matching'' $M$. 
\begin{lemma}
  \label{lemma1}
    For every clause $c$ in $\phi$, at least two of the three level~2 gadgets corresponding to $c$ have to be in unit state in $\vec{\alpha}$.
\end{lemma}
\begin{proof}
  Let $c$ be any clause in $\phi$.
  We know from observation~1 above that the level~3 gadget corresponding to $c$ is in {\em unit} state in $\vec{\alpha}$. 
  So $\alpha_{s_0} = \alpha_{t_0} = -1$.
  Also, one of the following three cases holds: (1)~$(s^c_0,t^c_1)$ and $(s^c_1,t^c_0)$ are in $M$, 
   (2)~$(s^c_0,t^c_2)$ and $(s^c_2,t^c_0)$ are in $M$,
   (3)~$(s^c_0,t^c_3)$ and $(s^c_3,t^c_0)$ are in $M$.

  \begin{itemize}
  \item In case~(1), the node $t_0$ prefers $p_4$ and $p_7$ to its partner $s_1$ in $M$. Thus $\wt_M(p_4,t_0) = \wt_M(p_7,t_0) = 0$.
    Since $\alpha_{t_0} = -1$, we need to have $\alpha_{p_4} = \alpha_{p_7} = 1$ so that $\alpha_{p_4} + \alpha_{t_0} \ge \wt_M(p_4,t_0)$ and
    $\alpha_{p_7} + \alpha_{t_0} \ge \wt_M(p_7,t_0)$. Thus the middle and rightmost level~2 gadgets corresponding to $c$
  (see Fig.~\ref{level2:example}) have to be in unit state in $\vec{\alpha}$.

  \item In case~(2), the node $t_0$ prefers $p_7$ to its partner $s_2$ in $M$ and the node $s_0$ prefers $q_0$ to its partner $t_2$ 
        in $M$. Thus $\alpha_{p_7} = \alpha_{q_0} = 1$ so that $\alpha_{p_7} + \alpha_{t_0} \ge \wt_M(p_7,t_0)$ and 
        $\alpha_{s_0} + \alpha_{q_0} \ge \wt_M(s_0,q_0)$. Thus the leftmost and rightmost level~2 gadgets corresponding to $c$ 
        (see Fig.~\ref{level2:example}) have to be in unit state in $\vec{\alpha}$.

  \item In case~(3), the node $s_0$ prefers $q_0$ and $q_3$ to its partner $t_3$ in $M$.
    Thus $\alpha_{q_0} = \alpha_{q_3} = 1$ so that $\alpha_{s_0} + \alpha_{q_0} \ge \wt_M(s_0,q_0)$ and 
    $\alpha_{s_0} + \alpha_{q_3} \ge \wt_M(s_0,q_3)$. Thus the leftmost and middle level~2 gadgets corresponding to $c$ 
    (see Fig.~\ref{level2:example}) have to be in unit state in $\vec{\alpha}$. \qed
  \end{itemize}
\end{proof}

\begin{lemma}
  \label{lemma2}
  For any clause $c$ in $\phi$, {\em at least one} of the level~1 gadgets
  corresponding to variables in $c$ is in unit state in $\vec{\alpha}$.
\end{lemma}
\begin{proof}
  We showed in Lemma~\ref{lemma1} that at least two of the three level~2 gadgets corresponding to $c$ are in unit state in $\vec{\alpha}$. 
  We have three cases here (see Fig.~\ref{level2:example}): (i)~the leftmost and middle gadgets are in unit state in $\vec{\alpha}$,
  (ii)~the leftmost and rightmost gadgets are in unit state in $\vec{\alpha}$, and 
  (iii)~the middle and rightmost gadgets are in unit state in $\vec{\alpha}$.

 Let us consider case~(i) first.  
  It follows from the proof of Lemma~\ref{lemma1} that $\alpha_{q_0} = \alpha_{q_3} = 1$. This also forces  
  $\alpha_{p_1} = \alpha_{p_4} = 1$. This is because $\alpha_{p_1}$ and $\alpha_{p_4}$ have to be non-negative since $p_1$ and $p_4$ are 
  neighbors of the unmatched node $z$. And so by Lemma~\ref{prop1}, $\alpha_{p_1} = 1$ and $\alpha_{p_4} = 1$.

  As $q_0$ and $p_1$ are the most preferred neighbors of $p_2$ and $q_2$ while $p_2$ and $q_2$ are the least preferred neighbors of $q_0$ 
  and $p_1$ in $M$ (by Theorem~\ref{thm:separate}), we have $\wt_M(p_2,q_0) = \wt_M(p_1,q_2) = 0$. 
  Since $(p_2,q_0)$ and $(p_1,q_2)$ are popular edges, it follows from Lemma~\ref{prop0} that $\alpha_{p_2} + \alpha_{q_0} = 0$ and
  $\alpha_{p_1} + \alpha_{q_2} = 0$. Thus $\alpha_{p_2} = \alpha_{q_2} = -1$ and so $(p_2,q_2) \notin M$. So either $(p_2,q_0),(p_0,q_2)$
  are in $M$ or $(p_2,q_1),(p_1,q_2)$ are in $M$ (by Theorem~\ref{thm:separate}). This means that either $\wt_M(p_2,y_j) = 0$ or 
  $\wt_M(x_k,q_2) = 0$. That is, either $\alpha_{y_j} = 1$ or $\alpha_{x_k} = 1$.

  Similarly, $\wt_M(p_5,q_3) = \wt_M(p_4,q_5) = 0$ and we can conclude that $\alpha_{p_5} = \alpha_{q_5} = -1$.
  Thus $(p_5,q_5) \notin M$ and either $(p_5,q_3),(p_3,q_5)$ are in $M$ or $(p_5,q_4),(p_4,q_5)$ are in $M$ (by Theorem~\ref{thm:separate}). 
  This means that either $\wt_M(p_5,y_k) = 0$ or $\wt_M(x_i,q_5) = 0$. That is, either $\alpha_{y_k} = 1$ or $\alpha_{x_i} = 1$. 
  Thus either (1)~the gadgets corresponding to variables $X_i$ and $X_j$ are in unit state or 
  (2)~the gadget corresponding to $X_k$ is in unit state in $\vec{\alpha}$.
  Thus in this case at least {\em one} of the level~1 gadgets corresponding to variables in $c$ is in unit state in $\vec{\alpha}$. 

  The proofs of case~(ii) and case~(iii) are quite similar. Let us consider case~(ii) next.  
  It follows from the proof of Lemma~\ref{lemma1} that $\alpha_{q_0} = \alpha_{p_7} = 1$.
  This also forces $\alpha_{p_1} = \alpha_{q_6} = 1$ and $\alpha_{p_2} = \alpha_{q_2} = -1$.  
  By the same reasoning as in case~(i), we have either $\wt_M(p_2,y_j) = 0$ or $\wt_M(x_k,q_2) = 0$. 
  That is, either $\alpha_{y_j} = 1$ or $\alpha_{x_k} = 1$. Similarly, $\alpha_{p_8} = \alpha_{q_8} = -1$ and 
  either $\wt_M(p_8,y_i) = 0$ or $\wt_M(x_j,q_8) = 0$, i.e., $\alpha_{y_i} = 1$ or $\alpha_{x_j} = 1$.

  So either (1)~the gadgets corresponding to variables $X_i$ and $X_k$ are in unit state or (2)~the gadget corresponding to $X_j$ 
  is in unit state in $\vec{\alpha}$. Thus in this case also at least one of the level~1 gadgets corresponding to variables in $c$ 
  is in unit state in $\vec{\alpha}$.

  In case~(iii), it follows from the proof of Lemma~\ref{lemma1} that $\alpha_{p_4} = \alpha_{p_7} = 1$.
  This forces $\alpha_{q_3} = \alpha_{q_6} = 1$ and $\alpha_{p_5} = \alpha_{q_5} = -1$. 
  By the same reasoning as in case~(i), we have either $\wt_M(p_5,y_k) = 0$ or $\wt_M(x_i,q_5) = 0$. 
  That is, either $\alpha_{y_k} = 1$ or $\alpha_{x_i} = 1$. Similarly, either $\alpha_{y_i} = 1$ or $\alpha_{x_j} = 1$.
  Thus either (1)~the gadgets corresponding to variables $X_j$ and $X_k$ are in unit state or (2)~the gadget corresponding to $X_i$ 
  is in unit state in $\vec{\alpha}$. Thus in this case also at least one of the level~1 gadgets corresponding to variables in $c$ 
  is in unit state in $\vec{\alpha}$. \qed
\end{proof}

\begin{lemma}
  \label{lemma3}
   For any clause $c$ in $\phi$, {\em at most one} of the level~1 gadgets corresponding to variables in $c$ is
  in unit state in $\vec{\alpha}$.
\end{lemma}
\begin{proof}
  We know from observation~2 made at the start of this section that all the three level~0 gadgets corresponding to $c$ are in zero state 
 in $\vec{\alpha}$. So $\alpha_{a_t} = \alpha_{b_t} = 0$ for $1 \le t \le 6$. We know from Theorem~\ref{thm:separate} that
  either $(a_1,b_1),(a_2,b_2)$ are in $M$ or $(a_1,b_2),(a_2,b_1)$ are in $M$. So either $\wt_M(a_1,y'_j) = 0$ or $\wt_M(x'_k,b_1) = 0$.
  So either $\alpha_{y'_j} \ge 0$ or $\alpha_{x'_k} \ge 0$.

  Consider any variable $X_r$. We know from Theorem~\ref{thm:separate} that either $\{(x_r,y'_r),(x'_r,y_r)\} \subseteq M$ or
  $\{(x_r,y_r),(x'_r,y'_r)\} \subseteq M$. It follows from Lemma~\ref{prop0} that $\alpha_{x_r} + \alpha_{y'_r} = \wt_M(x_r,y'_r) = 0$ and  
  $\alpha_{x'_r} + \alpha_{y_r} = \wt_M(x'_r,y_r) = 0$.
  Also  due to the nodes $z$ and  $z'$, we have $\alpha_{x_r} \ge 0$ and $\alpha_{y_r} \ge 0$. 
  Thus $\alpha_{y'_r} \le 0$ and $\alpha_{x'_r} \le 0$. 

  Hence we can conclude that either $\alpha_{y'_j} = 0$ or $\alpha_{x'_k} = 0$. In other words, either the gadget corresponding to $X_j$ 
  or the gadget corresponding to $X_k$ is in zero state. Similarly, by analyzing the level~0 gadget on nodes $a^c_t,b^c_t$ for 
  $t = 3,4$, we can show that either the gadget corresponding to $X_k$ or the gadget corresponding to $X_i$ is in zero state.
  Also, by analyzing the level~0 gadget on nodes $a^c_t,b^c_t$ for $t = 5,6$,  either the gadget corresponding to $X_i$ or the gadget 
  corresponding to $X_j$ is in zero state.

  Thus at least 2 of the 3 gadgets corresponding to variables in clause $c$ are in zero state in $\vec{\alpha}$. Hence at most 1 of 
  the 3 gadgets corresponding to variables in $c$ is in unit state in $\vec{\alpha}$. \qed
\end{proof}

\begin{theorem}
  \label{thm1}
If $G$ admits a desired popular matching then $\phi$ has a 1-in-3 satisfying assignment.
\end{theorem}
\begin{proof}
  Let $M$ be a desired popular matching in $G$. That is, $M$ matches all nodes except $z$ and $z'$. Let $\vec{\alpha} \in \{0,\pm 1\}^n$
  be a witness of $M$.
  
  We will now define a $\mathsf{true}$/$\mathsf{false}$ assignment for the variables in $\phi$.
  For each variable $X_r$ in $\phi$ do:
  \begin{itemize}
    \item If the level~1 gadget corresponding to $X_r$ is in {\em unit} state in $\vec{\alpha}$, i.e., if $\alpha_{x_r} = \alpha_{y_r} = 1$ 
      and $\alpha_{x'_r} = \alpha_{y'_r} = -1$ or equivalently, if $(x_r,y'_r)$ and $(x'_r,y_r)$ are in $M$, then set $X_r$ to $\mathsf{true}$.

      \smallskip
      
    \item Else set  $X_r$ to $\mathsf{false}$, i.e., the level~1 gadget corresponding to $X_r$ is in {\em zero} state in $\vec{\alpha}$ or 
          equivalently, $(x_r,y_r)$ and $(x'_r,y'_r)$ are in $M$.
  \end{itemize}
      
  Since $M$ is our desired popular matching, 
  it follows from Lemmas~\ref{lemma2} and \ref{lemma3} that for every clause $c$ in $\phi$, {\em exactly one} of the three level~1 gadgets
  corresponding to variables in $c$ is in unit state in $\vec{\alpha}$.
  When the gadget $X_r$ is in unit state, we have
  $\alpha_{x_r} = \alpha_{y_r} = 1$ and $\alpha_{x'_r} = \alpha_{y'_r} = -1$. This is due to the fact that 
  $\alpha_{x_r} \ge 0$ and $\alpha_{y_r} \ge 0$ due to the nodes $z$ and $z'$, respectively.

  Thus $X_r$ is in unit state in $\vec{\alpha}$ if and only if the edges $(x_r,y'_r)$ and $(x'_r,y_r)$ are in $M$. 
  Hence for each clause $c$ in $\phi$, exactly one of the
  variables in $c$ is set to $\mathsf{true}$. Hence this is a 1-in-3 satisfying assignment for $\phi$. \qed
\end{proof}

\subsection{The converse}
\label{converse}
Suppose $\phi$ admits a 1-in-3 satisfying assignment. We will now use this assignment to construct a desired popular matching $M$ in $G$.
The edge $(a_0,b_0)$ is in $M$. For each variable $X_r$ in $\phi$ do:
\begin{itemize}
\item if $X_r = \mathsf{true}$ then include the edges $(x_r,y'_r)$ and $(x'_r,y_r)$ in $M$;
\item else include the edges $(x_r,y_r)$ and $(x'_r,y'_r)$ in $M$.
\end{itemize}
  
Consider a clause $c = X_i \vee X_j \vee X_k$. We know that exactly one of $X_i,X_j,X_k$ is set to $\mathsf{true}$ in our assignment.
Assume without loss of generality that $X_j = \mathsf{true}$. 
We will include the following edges in $M$ from all the gadgets corresponding to $c$.
Corresponding to the level~0 gadgets for $c$ (see Fig.~\ref{level0:example}):
  \begin{itemize}
  \item  Add the edges  $(a^c_1,b^c_1), (a^c_2,b^c_2)$ from the leftmost gadget and $(a^c_5,b^c_6),(a^c_6,b^c_5)$
    from the rightmost gadget to $M$.

    We will select $(a^c_3,b^c_3),(a^c_4,b^c_4)$ from the middle gadget. (Note that we  could also have selected
    $(a^c_3,b^c_4),(a^c_4,b^c_3)$ from the middle gadget.) 
  \end{itemize}

Corresponding to the level~2 gadgets for $c$  (see Fig.~\ref{level2:example}):
  \begin{itemize}
  \item Add the edges $(p^c_0,q^c_0),(p^c_2,q^c_1),(p^c_1,q^c_2)$ from the leftmost gadget,
    $(p^c_3,q^c_3),(p^c_4,q^c_4),(p^c_5,q^c_5)$ from the middle gadget, and
    $(p^c_6,q^c_8),(p^c_7,q^c_7),(p^c_8,q^c_6)$ from the rightmost gadget to $M$.
  \end{itemize}


 Since the leftmost and rightmost level~2 gadgets (see Fig.~\ref{level2:example}) are dominant, we will include $(s^c_0,t^c_2)$ and $(s^c_2,t^c_0)$ in $M$. Hence
  \begin{itemize}
     \item Add the edges $(s^c_0,t^c_2), (s^c_1,t^c_1),(s^c_2,t^c_0),(s^c_3,t^c_3)$ to $M$.
  \end{itemize}

Thus $M$ matches all nodes except $z$ and $z'$. We will show the following theorem now.

\begin{theorem}
  The matching $M$ described above is a popular matching in $G$.
\end{theorem}
\begin{proof}
  We will prove $M$'s popularity by describing a witness $\vec{\alpha} \in \{0,\pm 1\}^n$. That is, $\sum_{u\in A\cup B} \alpha_u$ will be 0 and every edge will be
  covered by the sum of $\alpha$-values of its endpoints, i.e., $\alpha_u + \alpha_v \ge \wt_M(u,v)$ for all edges $(u,v)$ in $E$. We will also have
  $\alpha_u \ge \wt_M(u,u)$ for all nodes $u$.

  Set $\alpha_{a_0} = \alpha_{b_0} = \alpha_z = \alpha_{z'} = 0$. Also set $\alpha_u = 0$ for all nodes $u$ in the gadgets with {\em no} ``blocking edges''.
  This includes all level~0 gadgets,
  and the gadgets in level~1 that correspond to variables set to $\mathsf{false}$, and also the level~2 gadgets
  such as the gadget with nodes $p^c_3,q^c_3,p^c_4,q^c_4,p^c_5,q^c_5$ (the middle gadget in Fig.~\ref{level2:example}) since we assumed $X_j = \mathsf{true}$.

  \smallskip

  For every variable $X_r$ assigned to $\mathsf{true}$: set $\alpha_{x_r} = \alpha_{y_r} = 1$ and $\alpha_{x'_r} = \alpha_{y'_r} = -1$.
  For every clause, consider the level~2 gadgets corresponding to this clause with ``blocking edges'':
  for our clause $c$, these are the leftmost and rightmost gadgets in Fig.~\ref{level2:example} (since we assumed $X_j = \mathsf{true}$).

  Recall that we included in $M$ the edges $(p^c_0,q^c_0),(p^c_2,q^c_1),(p^c_1,q^c_2)$ from the leftmost gadget.
  We will set $\alpha_{q^c_0} = \alpha_{p^c_1} = \alpha_{q^c_1} = 1$ and $\alpha_{p^c_0} = \alpha_{p^c_2} = \alpha_{q^c_2} = -1$.
  We also included in $M$ the edges $(p^c_6,q^c_8),(p^c_7,q^c_7),(p^c_8,q^c_6)$ from the rightmost gadget.
  We will set $\alpha_{p^c_6} = \alpha_{q^c_6} = \alpha_{p^c_7} = 1$ and $\alpha_{q^c_7} = \alpha_{p^c_8} = \alpha_{q^c_8} = -1$.

  In the level~3 gadget corresponding to $c$, we included in $M$ the edges $(s^c_0,t^c_2), (s^c_1,t^c_1),(s^c_2,t^c_0)$, $(s^c_3,t^c_3)$.
  We will set $\alpha_{t^c_1} = \alpha_{s^c_2} = \alpha_{t^c_2} = \alpha_{s^c_3} = 1$ and $\alpha_{s^c_0} = \alpha_{t^c_0} = \alpha_{s^c_1} = \alpha_{t^c_3} = -1$.

  The claim below shows that $\vec{\alpha}$ is indeed a valid witness for $M$. Thus $M$ is a popular matching. \qed
 \end{proof}

 \begin{new-claim}
   The vector $\vec{\alpha}$ defined above is a witness to $M$'s popularity.
 \end{new-claim}
 \begin{proof}
   For any edge $(u,v) \in M$, we have $\alpha_u + \alpha_v = 0$, also $\alpha_z = \alpha_{z'} = 0$. Thus $\sum_{u \in A \cup B}\alpha_u = 0$.
   For any neighbor $v$ of $z$ or $z'$, we have
  $\alpha_v \ge 0$. Thus all edges incident to $z$ or $z'$ are covered by the sum of $\alpha$-values of their endpoints.
  It is also easy to see that for every intra-gadget edge $(u,v)$, we have $\alpha_u + \alpha_v \ge \wt_M(u,v)$.

  In particular, the endpoints of
  every blocking edge to $M$ have their $\alpha$-value set to 1. When $X_j = \mathsf{true}$, the edge
  $(x_j,y_j)$ is a blocking edge to $M$ and so are $(p^c_1,q^c_1),(p^c_6,q^c_6),(s^c_2,t^c_2)$ in the gadgets involving clause $c$. 
  We will now check that the edge covering constraint holds for all edges $(u,v)$  where $u$ and $v$ belong to different levels.

  \begin{itemize}
  \item  Consider edges in $G$ between a
    level~0 gadget and a level~1 gadget. When $X_j = \mathsf{true}$, the edges $(a^c_1,y'_j)$ and $(x'_j,b^c_5)$ are the most interesting as they have one endpoint
    in a gadget with $\alpha$-values 0 and another endpoint in a gadget with $\alpha$-values equal to $\pm 1$.

  Observe that both these edges are labeled $(-,-)$. This is because $a^c_1$ prefers its partner $b^c_1$
  to $y'_j$ and symmetrically, $y'_j$ prefers its partner $x_j$ to $a^c_1$. Thus $\wt_M(a^c_1,y'_j) = -2 < \alpha_{a^c_1} + \alpha_{y'_j}  = 0 - 1$.
  Similarly, $b^c_5$ prefers its partner $a^c_6$ to $x'_j$ and symmetrically, $x'_j$ prefers its partner $y_j$ to $b^c_5$.
  Thus  $\wt_M(x'_j,b^c_5) = -2 < \alpha_{x'_j}  + \alpha_{b^c_5} = - 1 + 0$.

  \smallskip
   
   \item We will now consider edges in $G$ between a level~1 gadget and a level~2 gadget. 
   We have $\wt_M(p^c_2,y_j) = 0$ since $p^c_2$ prefers $y_j$ to its partner $q^c_1$ while  $y_j$ prefers its partner $x'_j$ to $p^c_2$. We have
   $\alpha_{p^c_2} + \alpha_{y_j} = -1 + 1 = \wt_M(p^c_2,y_j) = 0$. The edge $(x_k,q^c_2)$ is labeled $(-,-)$ and we have $\alpha_{x_k} = 0$ and
   $\alpha_{q_2} = -1$.
   
   Similarly, the edge $(p^c_8,y_i)$ is labeled $(-,-)$ and so this is covered by the sum of $\alpha$-values of its endpoints.
   We have $\wt_M(x_j,q^c_8) = 0 = 1 - 1 = \alpha_{x_j} + \alpha_{q^c_8}$.
   We also have $\wt_M(p^c_5,y_k) = 0$ and $\alpha_{p^c_5} = \alpha_{y_k} = 0$. Similarly, 
   $\wt_M(x_i,q^c_5) = 0$ and $\alpha_{x_i} = \alpha_{q^c_5} =  0$. Thus all these edges are covered.

  \smallskip
   
\item We will now consider edges in $G$ between a level~2 gadget and a level~3 gadget. First, consider the edges
  $(s^c_0,q^c_0), (s^c_0,q^c_3), (p^c_7,t^c_0), (p^c_4,t^c_0)$.
   We have $\wt_M(s^c_0,q^c_0) = 0$ and $\alpha_{s^c_0} = -1, \alpha_{q^c_0} = 1$, so this edge is covered. Similarly, $\wt_M(p^c_7,t^c_0) = 0$ and
   $\alpha_{p^c_7} = 1,\alpha_{t^c_0} = -1$.  The edges $(s^c_0,q^c_3)$ and $(p^c_4,t^c_0)$ are labeled $(-,-)$, so they are also covered.  
   Next consider the edges $(s^c_0,q^{c_i}_{3j+2})$ and $(p^{c_i}_{3j+2},t^c_0)$ for any clause $c_i$ and $j \in \{0,1,2\}$.
   It is easy to see that these edges are labeled $(-,-)$, so these edges are also covered.

   \smallskip

 \item Finally consider the edges between a level~1 gadget and a level~3 gadget. Corresponding to clause $c$, these edges are $(s^c_0,y'_i)$ and $(x'_i,t^c_0)$
   for any $i \in [n_0]$.
   It is again easy to see that these edges are labeled $(-,-)$ and so they are covered. Thus it follows that $\vec{\alpha}$ is a witness to $M$'s popularity. \hfill $\lozenge$
   \end{itemize}
\end{proof}

\subsection{Proof of Theorem~\ref{thm:separate}}
\label{sec:proof-separate}

Let $c = X_i \vee X_j \vee X_k$ be a clause in $\phi$. We will show in the following claims that no edge between 2 different gadgets can be popular.

\begin{new-claim}
  No edge between a level~0 node and a level~1 node can be popular.
\end{new-claim}
\begin{proof}
  Consider any such edge  in $G$, say $(a^c_1,y'_j)$. In order to show this edge cannot be present in any popular matching, we will show a popular matching $S$ along
  with a witness $\vec{\alpha}$ such that $\alpha_{a^c_1} + \alpha_{y'_j} > \wt_S(a^c_1,y'_j)$. Then it will immediately follow from the {\em slackness} of this edge that
  $(a^c_1,y'_j)$ is not used in any popular matching (by Lemma~\ref{prop0}).

  Let $S$ be the men-optimal stable matching in $G$. The vector $\vec{\alpha} = \vec{0}$ is a witness to $S$. The edges $(a^c_1,b^c_1)$
  and $(x'_j,y'_j)$ belong to $S$, so we have $\wt(a^c_1,y'_j) = -2$ while $\alpha_{a^c_1} = \alpha_{y'_j} = 0$. Thus $(a^c_1,y'_j)$ is not a popular edge.
  We can similarly show that $(x'_k,b^c_1)$ is not a popular edge by considering the women-optimal stable matching $S'$. \hfill $\lozenge$
\end{proof}

\begin{new-claim}
  No edge between a level~1 node and a level~2 node is popular.
\end{new-claim}
\begin{proof}
  Consider any such edge  in $G$, say $(p^c_2,y_j)$.
  Consider the dominant matching $M^*$ that contains the edges $(p^c_0,q^c_2),(p^c_2,q^c_0)$ and also
  $(x_j,y'_j),(x'_j,y_j)$. Note that $\wt_{M^*}(p^c_2,y_j) = -2$.

  A witness $\vec{\beta}$ to $M^*$ sets $\beta_{p^c_2} = \beta_{q^c_2} = -1$ and $\beta_{x_j} = \beta_{y_j} = 1$. This is because $(x_j,y_j)$ and $(p^c_0,q^c_0)$
  are blocking edges to $M^*$, so $\beta_{x_j} = \beta_{y_j} = 1$ and similarly, $\beta_{p^c_0} = \beta_{q^c_0} = 1$ (this makes $\beta_{p^c_2} = \beta_{q^c_2} = -1$).
  So $\beta_{p^c_2} + \beta_{y_j} = 0$ while $\wt_{M^*}(p^c_2,y_j) = -2$. Thus this edge
  is slack and so it cannot be a popular edge.

  We can similarly show that the edge $(x_k,q^c_2)$ is not popular by considering the dominant matching
  that includes the edges $(p^c_1,q^c_2)$ and $(p^c_2,q^c_1)$. \hfill $\lozenge$
\end{proof}

\begin{new-claim}
  No edge between a node in level~3  and a node in levels~1 or 2 is popular.
\end{new-claim}  
\begin{proof}
  Consider the edge $(s^c_0,q^c_0)$. The dominant matching $M^*$ includes the edges $(s^c_0,t^c_1),(s^c_1,t^c_0)$, and $(p^c_2,q^c_0)$. 
  So we have $\wt_{M^*}(s^c_0,q^c_0) = -2$ while $\beta_{s^c_0} = -1$ and $\beta_{q^c_0} = 1$, where
  $\vec{\beta}$ is a witness to $M^*$. Hence $(s^c_0,q^c_0)$ is not a popular edge. It can similarly be shown for any edge
  $e \in \{(s^c_0,q^c_3),(p^c_4,t^c_0),(p^c_7,t^c_0)\}$ that $e$ is not a popular edge.

  \smallskip
  
  Suppose the edge $(s^c_0,u)$ for some $u \in \{y'_i: i\in[n_0]\} \cup \{q^{c_i}_2,q^{c_i}_5,q^{c_i}_8: c_i$ is a clause$\}$ belongs to a popular matching $M$.
  In order to show a contradiction,
  consider the edge $(s^c_0,t^c_1)$. We have $\wt_M(s^c_0,t^c_1) = 0$ since $s^c_0$ prefers $t^c_1$ to $u$ while $t^c_1$ prefers its partner in $M$
  (this is $s^c_1$) to $s^c_0$. Since $\alpha_{s^c_0} = -1$, it has to be the case that $\alpha_{t^c_1} = 1$.

  Since $s^c_0$ is matched to a node outside $\{t^c_1,t^c_2,t^c_3\}$, the node $t^c_0$ also has to be matched to a node outside $\{s^c_1,s^c_2,s^c_3\}$ ---
  otherwise one of the 3 stable nodes $t^c_1,t^c_2,t^c_3$ would be left unmatched in $M$; however as $M$ is popular, every stable node 
has to be matched in $M$. We have also seen
  earlier that neither $(p^c_4,t^c_0)$ nor $(p^c_7,t^c_0)$ is popular. Thus $t^c_0$ has to be matched to a neighbor worse than $s^c_1$.
  
  Thus $\wt_M(s^c_1,t^c_0) = 0$ and so $\alpha_{s^c_1} = 1$. Since $(s^c_1,t^c_1)$ is a stable edge, it follows from Lemma~\ref{prop0} that
  $\alpha_{s^c_1} + \alpha_{t^c_1} = \wt_M(s^c_1,t^c_1)$. However $\wt_M(s^c_1,t^c_1) = 0$ since $(s^c_1,t^c_1) \in M$ and we have just shown that
  $\alpha_{s^c_1} = \alpha_{t^c_1} = 1$. This is a contradiction and thus $(s^c_0,u) \notin M$. \hfill $\lozenge$
\end{proof}

It is easy to see that the nodes $z$ and $z'$ are in the same connected component of $F_G$ as the dominant matching $M^*$ contains the edges
$(a_0,z)$ and $(z',b_0)$ while any stable matching in $G$ contains $(a_0,b_0)$.
We will now show that any popular matching that matches $z$ and $z'$ has to match these nodes to $a_0$ and $b_0$, respectively.

\begin{lemma}
  \label{bipartite:lemma1}
      If $M$ is a popular matching in $M$ that matches $z$ and $z'$ then $\{(a_0,z),(z',b_0)\} \subseteq M$.
\end{lemma}
\begin{proof}
Suppose $(x_i,z) \in M$ for some $i \in [n_0]$. We know from the above claims that there is no popular edge
between $x_i$'s gadget and any neighbor in levels~0, 2, or 3. So $(x_i,z) \in M$ implies that $(z',y_i) \in M$ since all the 4 nodes $x_i,y_i,x'_i,y'_i$
have to be matched in $M$ and there is no other possibility of a popular edge incident to either $y_i$ or $y'_i$. Hence $(x'_i,y'_i) \in M$.
Thus $(x_i,y'_i)$ and $(x'_i,y_i)$ are blocking edges to $M$.

This means that $\alpha_{x_i} = \alpha_{y_i} = \alpha_{x'_i} = \alpha_{y'_i} = 1$. Note that $(x'_i,y'_i)$ is a stable edge and so
$\alpha_{x'_i} + \alpha_{y'_i} = \wt_M(x'_i,y'_i)$. However $\wt_M(x'_i,y'_i) = 0$ while $\alpha_{x'_i} = \alpha_{y'_i} = 1$. This is a contradiction and hence
$(x_i,z) \notin M$ for any $i \in [n_0]$. We can similarly show that neither $(p^c_{3j+1},z)$ nor $(z',q^c_{3j})$ is in $M$, for any $j \in \{0,1,2\}$ and any clause $c$.

\smallskip

So if $z$ and $z'$ are matched in $M$ then it has to be either with $a_0,b_0$ or with some level~0 neighbors.
Observe that if $(a^c_{2i-1},z) \in M$ then $(z',b^c_{2i-1}) \in M$ as the 4 nodes $a^c_{2i-1},b^c_{2i-1},a^c_{2i},b^c_{2i}$ have to be matched in $M$
and there is no other possibility of a popular edge incident to any of these 4 nodes (by our first claim in this section).

Suppose $(a^c_{2i-1},z)$ and $(z',b^c_{2i-1})$ are in $M$ for some $c$ and $i \in \{1,2,3\}$. Since $z$ prefers $a_0$ to any level~0
neighbor, we have $\wt_M(a_0,z) = 0$. Similarly, $\wt_M(z',b_0) = 0$. Since $\wt_M(a_{2i-1}^c, b_{2i-1}^c) = 2$, this implies that $\alpha_{a_{2i-1}^c} = \alpha_{b_{2i-1}^c} = 1$. Hence, $\alpha_z = \alpha_{z'} = -1$ and we have
$\alpha_{a_0} = \alpha_{b_0} = 1$. Note that $\alpha_{a_0} + \alpha_{b_0}$ has to be equal to $\wt_M(a_0,b_0)$ as $(a_0,b_0)$
is a popular edge.

If $a_0$ is not matched to $z$ (and so $b_0$ is not matched to $z'$), then 
$(a_0,b_0) \in M$. So $\wt_M(a_0,b_0) = 0$, however $\alpha_{a_0} + \alpha_{b_0} = 2$, a contradiction.
Thus if $M$ matches $z$ and $z'$ then $\{(a_0,z),(z',b_0)\} \subseteq M$. \qed
\end{proof}

Thus the above claims and Lemma~\ref{bipartite:lemma1} show that every level~$i$ gadget (for $i = 0,1,2,3$) forms a distinct connected component
in the graph $F_G$ and the 4 nodes $a_0,b_0,z$, and $z'$ belong to their own connected component.
This finishes the proof of Theorem~\ref{thm:separate}.

We have shown a polynomial time reduction from 1-in-3 SAT to the problem of deciding if $G = (A \cup B, E)$ admits a popular matching that 
matches all nodes except $z$ and $z'$. That is, we have shown the following theorem.

 \begin{theorem}
   \label{thm:section3}
   The instance $G = (A \cup B, E)$ admits a popular matching that matches all nodes except $z$ and $z'$ if and only if $\phi$ is in 1-in-3 SAT.
 \end{theorem}  

 Theorem~\ref{bipartite:theorem} showed that $G$ admits a popular matching that is neither stable nor dominant if and only if $G$ admits a popular matching
 that matches all nodes except $z$ and $z'$. Hence Theorem~\ref{final-thm} follows from Theorem~\ref{thm:section3}.

\section{Related hardness results}\label{sec:consequences}

\subsection{Popular matchings with forced/forbidden elements}\label{sec:forbidden-forced}

We now consider the \emph{popular matching problem in $G = (A \cup B, E)$ with forced/forbidden elements} (\texttt{pmffe}), which is the variant of the popular matching problem in a bipartite graph $G = (A \cup B, E)$, where our input also consists of some forced (resp., forbidden) edges $E_1$ (resp., $E_0$), and/or some forced (resp. forbidden) nodes $U_1$ (resp. $U_0$). The goal is to compute a popular matching $M$ in $G$ where all forced elements are included in $M$ 
and no forbidden element is included in $M$. 
The following result is immediate.

\begin{corollary}
  \label{cor:forced-forb1}
  The popular matching problem in $G = (A \cup B, E)$ with a given forced node set $U_1$ and a given forbidden node set $U_0$ is NP-hard for $|U_0| = |U_1| = 1$.
\end{corollary}  

The proof of Corollary~\ref{cor:forced-forb1} follows from our instance $G$ in Section~\ref{sec:hardness} with $U_0 = \{z\}$ and $U_1 = \{s^c_0\}$. It follows from the proof of
Lemma~\ref{bipartite:lemma3} that the nodes $s^c_0,t^c_0$ for all clauses $c$ have to be matched in such a popular matching $M$. So $M$ has to match all nodes
in $G$ except $z$ and $z'$. Theorem~\ref{thm:section3} showed that finding such a popular matching $M$ is NP-hard.

The variant of the popular matching problem in $G = (A \cup B, E)$ with a given forced node set $U_1$ where $|U_1| = 1$ and a given forced/forbidden edge set
$E_0 \cup E_1$ where $|E_0 \cup E_1| = 1$ is also NP-hard. To show this, consider $U_1 = \{s^c_0\}$ and either $E_1 = \{(a_0,b_0)\}$ or $E_0 = \{(a_0,z)\}$.
The forced node set forces the nodes $s^c_0,t^c_0$ for all clauses $c$ to be matched in our popular matching $M$ while $E_1$ (similarly, $E_0$)
forces $z$ and $z'$ to be unmatched in $M$.

In order to show the NP-hardness of the variant with $|E_0| = 2$ or $|E_1| = 2$, we will augment our instance $G$ in Section~\ref{sec:hardness} with an extra
level~1 gadget $X_0$ (see Fig.~\ref{level1:example}). Call the new instance $G_0$. The gadget $X_0$ has 4 nodes $x_0,x'_0,y_0,y'_0$ with the following preferences:

\begin{minipage}[c]{0.45\textwidth}
			
			\centering
			\begin{align*}
			        &x_0\colon \, y_0 \succ y'_0    \qquad\qquad &&  y_0\colon \, x_0 \succ x'_0  \\
                                &x'_0\colon \, y_0 \succ y'_0 \succ t^{c_1}_0 \succ \cdots \succ t^{c_k}_0 \qquad\qquad &&  y'_0\colon \, x_0 \succ x'_0 \succ s^{c_1}_0 \succ \cdots \succ s^{c_k}_0\\
			\end{align*}
\end{minipage}

Thus nodes in the gadget $X_0$ are not adjacent to $z$ or $z'$ or to any node in levels~0, 1, or 2 --- however $x'_0$ is adjacent to $t^c_0$ for all clauses $c$
and $y'_0$ is adjacent to $s^c_0$ for all clauses $c$. For each clause $c$, the node $y'_0$ is at the bottom of $s^c_0$'s 
preference list and the node $x'_0$ is at the bottom of $t^c_0$'s preference list.

A stable matching in the instance $G_0$ includes the edges $(x_0,y_0)$ and $(x'_0,y'_0)$ while there is a dominant matching in this instance
with the edges $(x_0,y'_0)$ and $(x'_0,y_0)$. It is easy to extend Theorem~\ref{thm:separate} to show that the popular subgraph for the 
instance $G_0$ is the popular subgraph $F_G$ along with an extra connected component with 4 nodes $x_0,y_0,x'_0,y'_0$. 
We are now ready to show the following result.

\begin{lemma}
  \label{lem:forced-forb2}
  The popular matching problem in $G = (A \cup B, E)$ with a given forced edge set $E_1$ is NP-hard for $|E_1| = 2$.
\end{lemma}  
 
The proof of Lemma~\ref{lem:forced-forb2} follows from the instance $G_0$ with $E_1 = \{(a_0,b_0), (x_0,y'_0)\}$.
Let $M$ be a popular matching in $G_0$ that includes the edges $(a_0,b_0)$ and $(x_0,y'_0)$.
Since $M$  has to contain $(a_0,b_0)$, it means that the nodes $z$ and $z'$ are unmatched in $M$.

Since $(x_0,y'_0) \in M$, it means that
$(x'_0,y_0) \in M$. So $(x_0,y_0)$ is a blocking edge to $M$ and we have $\alpha_{x_0} = \alpha_{y_0} = 1$ and  $\alpha_{x'_0} = \alpha_{y'_0} = -1$.
This forces $s^c_0,t^c_0$ for all $c$ to be matched in $M$. The argument is the same as in Lemma~\ref{bipartite:lemma3} since the edges $(s^c_0,y'_0)$
and  $(x'_0,t^c_0)$ would not be covered otherwise. This is because $\alpha_{x'_0} = \alpha_{y'_0} = -1$ and if $s^c_0,t^c_0$ are unmatched then
$\alpha_{s^c_0} = \alpha_{t^c_0} = 0$ and $\wt_M(s^c_0,y'_0) = \wt_M(x'_0,t^c_0) = 0$. This would make $\alpha_{s^c_0} + \alpha_{y'_0} < \wt_M(s^c_0,y'_0)$
and similarly,  $\alpha_{x'_0} + \alpha_{t^c_0}  < \wt_M(x'_0,t^c_0)$.

\smallskip

Thus $M$ has to be a popular matching that matches all nodes in $G_0$ except $z$ and $z'$. It is easy to see that the proof of Theorem~\ref{thm:section3} implies
that finding such a popular matching in $G_0$ is NP-hard.

We can similarly show that the popular matching problem in $G_0$ with a given forbidden edge set $E_0$ is NP-hard for $|E_0| = 2$.
For this, we will take $E_0 = \{(a_0,z),(x_0,y_0)\}$. This is equivalent to setting $E_1 = \{(a_0,b_0),(x_0,y'_0)\}$.

We can similarly show that this problem with a given forced set $E_1$ and a given forbidden edge set $E_0$ is NP-hard for $|E_0| = |E_1| = 1$.
For this, we will take $E_0 = \{(a_0,z)\}$ and $E_1 = \{(x_0,y'_0)\}$. This will force $s^c_0,t^c_0$ for all $c$ to be matched in $M$ while $z,z'$ are unmatched in $M$.

Finally, the variant with $|U_0| = 1$ and $|E_0\cup E_1| = 1$ follows by taking $U_0 = \{z\}$ and $E_0 = \{(x_0,y_0)\}$ or $E_1 = \{(x_0,y'_0)\}$.

\medskip

We now put all those observations together in the following theorem.

\begin{theorem}
  \label{thm:forced-forb3}
  The popular matching problem in $G = (A \cup B, E)$ with forced/forbidden element set $\langle E_0, E_1, U_0, U_1\rangle$ is NP-hard
when (i)~$|E_0| = 2$, (ii)~$|E_1| = 2$, (iii)~$|E_0| = |E_1| = 1$, (iv)~$|U_0| = |U_1| = 1$, (v)~$|U_0| = 1$ and $|E_0\cup E_1| = 1$,
and (vi)~$|U_1| = 1$ and $|E_0\cup E_1| = 1$.
\end{theorem}  

Note that when $|E_1|=1$ and $E_0=U_0=U_1=\emptyset$, \texttt{pmffe} reduces to the \emph{popular edge problem}, that can be solved in polynomial time~\cite{CK16}. Now suppose $|E_0|=1$ and $E_1=U_0=U_1=\emptyset$: we show below that a polynomial-time algorithm for \texttt{pmffe} follows from the algorithm for the popular edge problem. Define $E_s$ and $\overline{E}_s$ as:
$$\begin{array}{lll}
E_s &= \{ e \in E: \exists \text{ stable matching } M \text{ s.t. } e \in M\},\\
\overline{E}_s &= \{ e \in E: \exists \text{ stable matching } M \text{ s.t. } e \notin M\}.
\end{array}$$
$E_d$, $\overline{E}_d$ (resp., ${E_p}$, $\overline{E}_p$) are defined similarly, by replacing ``stable'' with ``dominant'' (resp., popular). It was proved in~\cite{CK16} that $E_p = E_s \cup E_d$. We now argue that $\overline{E}_p = \overline{E}_s \cup \overline{E}_d$.

Consider any $e \in \overline{E}_s \cup \overline{E}_d$. Since there is a stable or dominant matching that does not contain $e$, it follows that  $e \in \overline{E}_p$. We conclude that $\overline{E}_p \supseteq \overline{E}_s \cup \overline{E}_d$. We will now show that $\overline{E}_p \subseteq \overline{E}_s \cup \overline{E}_d$.

Consider any $e=(i,j) \in \overline{E}_p$. There is a popular matching $M$ s.t. $e \notin M$, and $i$ or $j$ is matched (if $i$ and $j$ are unmatched then $(i,j) = (+,+)$, and $M$ is not popular). Without loss of generality assume $(j,k) \in M$. It follows from $E_p = E_s \cup E_d$ that there exists a stable or dominant matching $M'$ s.t. $(j,k) \in M'$, hence $(i,j) \notin M'$. Thus $\overline{E}_p \subseteq \overline{E}_s \cup \overline{E}_d$.

Since $\overline{E}_p = \overline{E}_s \cup \overline{E}_d$, we can solve the forbidden edge problem by checking if $e \in \overline{E}_s$ or $e \in \overline{E}_d$. For stable matchings, this can be done in polynomial time, see e.g.~\cite{GI}. For dominant matchings, it immediately follows from results from~\cite{CK16}.

\subsection{Weighted popular matching problems} 

In this section, we are given, together with the usual bipartite graph and the rankings, a nonnegative edge weight vector $c\geq 0$. First, consider the problem of finding a popular matching of \emph{minimum} weight wrt $c$ (\texttt{min-wp}). Recall that \texttt{pmffe} with $|E_0|=2$, $U_1=U_0=E_1=\emptyset$ is NP-hard, as shown by the instance with graph $G_0$ and $E_0=\{(a_0,z),(x_0,y_0)\}$ (see Section \ref{sec:forbidden-forced}). Let the weights of $(a_0,z)$ and $(x_0,y_0)$ be equal to $1$, and let all other weights be $0$. We can conclude the following.

\begin{corollary}\label{cor:miwp}
\texttt{min-wp} is NP-Hard.
\end{corollary}

Now consider the problem of finding a popular matching of \emph{maximum} weight wrt $c$. We denote this problem by \texttt{max-wp}. It was shown in \cite{Kav18} that this problem is NP-hard. We show here a tight inapproximability result.

\begin{theorem}\label{thr:mwp}
  Unless P=NP, \texttt{max-wp} cannot be approximated in polynomial time to a factor better than $\frac{1}{2}$. On the other hand, there is a polynomial-time algorithm that computes a $\frac{1}{2}$-approximation to \texttt{max-wp}.
\end{theorem}

\begin{proof}
Consider the instance $G_0$ again. The hardness of approximating \texttt{max-wp} to a factor better than $\frac{1}{2}$ immediately follows from Lemma~\ref{lem:forced-forb2} by setting the weights of edges $(a_0,b_0)$ and $(x_0,y_0')$ to $1$ and all other edge weights to $0$. 

We will now show that a popular matching in $G$ of weight at least $c(M^*)/2$ can be computed in polynomial time, where
$M^*$ is a max-weight popular matching in $G$.
It was shown in \cite{CK16} that  any popular matching $M$ in $G = (A \cup B, E)$ can be 
partitioned into $M_0 \cup M_1$ such that $M_0 \subseteq S$ and $M_1 \subseteq D$, 
where $S$ is a stable matching and $D$ is a dominant matching in $G$.

Consider the following algorithm.

\begin{enumerate}
\item Compute a max-weight stable matching $S^*$ in $G$.
\item Compute a max-weight dominant matching $D^*$ in $G$.
\item Return the matching in $\{S^*, D^*\}$ with larger weight.
\end{enumerate}

Since all edge weights are non-negative, either the max-weight stable matching in $G$ or the max-weight dominant matching in $G$ has weight at least $c(M^*)/2$. 
Thus Steps~1-3 compute a $\frac{1}{2}$-approximation for max-weight popular matching in $G = (A \cup B,E)$.

Regarding the implementation of this algorithm, both $S^*$ and $D^*$ can be computed in polynomial time~\cite{Rot92,CK16}.
Thus our algorithm runs in polynomial time. \qed
\end{proof}

\subsection{Popular and Dominant matchings in non-bipartite graphs}

In this section, we consider the popular (similarly, dominant) problem in general graphs. The only modification to the input with respect to the previous paragraphs, is that our input graph $G$ is not required to be bipartite. As usual, together with $G$, we are given a collection of rankings, one per node of $G$, with each node ranking its neighbors in a strict order of preference. The goal is to decide if $G$ admits a popular (resp., dominant) matching.

Theorem~\ref{thm:section3} showed that that the problem of deciding if a bipartite instance $G = (A \cup B, E)$ admits a popular matching
that matches {\em exclusively} a given set $S \subset A \cup B$ is NP-hard (i.e., the nodes outside $S$ have to be left unmatched).
Let us call this the {\em exclusive popular set} problem.
We will now use the hardness of the exclusive popular set problem in the instance $G$ from Section~\ref{sec:hardness}
to show that the dominant problem in non-bipartite graphs is NP-hard. In order to show this, {\em merge} the nodes $z$ and $z'$ in the instance $G$
from Section~\ref{sec:hardness} into a single node $z$. Call the new graph $G'$.

The preference list of the node $z$ in $G'$ is all its level~1 neighbors in some order of preference, followed by
all its level~2 neighbors in some order of preference, followed by $a_0,b_0$, and then level~0 neighbors. The order among level~$i$ neighbors (for $i = 0,1,2$)
in this list does not matter.

\begin{lemma}
     A popular matching $N$ in $G'$ is dominant if and only $N$ matches all nodes in $G'$ except $z$.
  \end{lemma}
  \begin{proof}
    Let $N$ be any popular matching in $G'$. Any popular matching has to match all stable nodes in $G'$~\cite{HK11},
    thus $N$ matches all stable nodes in $G'$. Suppose some unstable node other than $z$ (say, $s^c_0$) is left unmatched in $N$. 
    We claim that $t^c_0$ also has to be left unmatched in $N$. Since $s^c_1$ and $t^c_1$ have no other neighbors, the edge $(s^c_1,t^c_1) \in N$
    and so there is an augmenting path
    $\rho = s^c_0$-$t^c_1$-$s^c_1$-$t^c_0$ with respect to $N$. Observe that $N$ is {\em not} more popular than $N \oplus \rho$, a larger matching.
    Thus $N$ is not a dominant matching in $G'$.
    
    In order to justify that $t^c_0$ also has to be left unmatched in $N$,
    let us view $N$ as a popular matching in $G$.
    We know that  $s^c_0$ and $t^c_0$ belong to the same connected component in the popular
    subgraph $F_G$ (by Lemma~\ref{lem:conn-comp}). So if $s^c_0$ is left unmatched in $N$, then $t^c_0$ is also unmatched in $N$ (by Lemma~\ref{prop1}).

    Conversely, suppose $N$ is a popular matching in $G'$ that matches all  nodes except $z$. Then there is no larger matching than $N$ in $G'$,
    thus $N$ is a dominant matching in $G'$. \qed
  \end{proof}
  
  Thus a dominant matching exists in $G'$ if and only if there is a popular matching in $G'$ that matches all nodes except $z$.
  This is equivalent to deciding if there exists a popular matching in $G$ that matches all nodes in $G$ except $z$ and $z'$. Thus we 
  have shown the following theorem.

  \begin{theorem}
  \label{second-thm}
  Given a graph $G = (V,E)$ with strict preference lists, the problem of deciding if $G$ admits a dominant matching or not
  is NP-hard. Moreover, this hardness holds even when $G$ admits a stable matching.
\end{theorem}  

  We will now show that the popular matching problem in non-bipartite graphs is also NP-hard.
  For this, we will augment the graph $G'$ with the gadget $D$ given in Fig.~\ref{D:example}. Call the new graph $H$.

\begin{figure}[h]
\centerline{\resizebox{0.16\textwidth}{!}{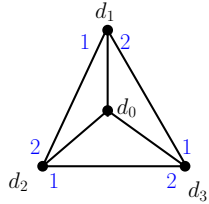}}
\caption{Each of $d_1,d_2,d_3$ is a top choice neighbor for another node here and $d_0$ is the last choice of $d_1,d_2,d_3$.}
\label{D:example}
\end{figure}

\medskip

\noindent{\em The gadget $D$.}
There will be 4 nodes $d_0,d_1,d_2,d_3$ that form the gadget $D$ (see  Fig.~\ref{D:example}).
The preferences of nodes in $D$ are given below.

\begin{minipage}[c]{0.45\textwidth}
			
			\centering
			\begin{align*}
				&d_1\colon \, d_2  \succ d_3 \succ d_0 \qquad\qquad && d_2\colon \, d_3  \succ d_1 \succ d_0\\
			        &d_3\colon \, d_1 \succ d_2 \succ d_0  \qquad\qquad && d_0\colon \, d_1 \succ d_2 \succ d_3 \succ \cdots \\
			\end{align*}
\end{minipage}

The node $d_0$ will be adjacent to all nodes in $H$, except $z$. The neighbors of $d_0$ that are not in $D$ are in the ``$\cdots$'' 
part of $d_0$'s preference list and the
order among these nodes does not matter. The node $d_0$ will be at the bottom of preference lists of  all its neighbors.

\begin{lemma}
    \label{new-lemma1}
    For any popular matching $M$ in $H$, the following properties hold:
    \begin{itemize}
    \item[(1)] either $\{(d_0,d_1), (d_2,d_3)\} \subset M$ or $\{(d_0,d_2), (d_1,d_3)\} \subset M$.
    \item[(2)] $M$ matches all nodes in $H$ except $z$.
    \end{itemize}  
\end{lemma}      
\begin{proof}
  Since each of $d_1,d_2,d_3$ is a top choice neighbor for some node in $H$, a popular matching in $H$ cannot leave any of these 3 nodes unmatched.
  Since these 3 nodes have no neighbors outside themselves other than $d_0$,
  a popular matching has to match $d_0$ to one of  $d_1,d_2,d_3$.
  Thus $d_0,d_1,d_2,d_3$ are matched among themselves in $M$.

  The only possibilities for $M$ when restricted to $d_0,d_1,d_2,d_3$ are the pair of edges $(d_0,d_1), (d_2,d_3)$ and $(d_0,d_2), (d_1,d_3)$.
  The third possibility
  $(d_0,d_3),(d_1,d_2)$ is ``less popular than''  $(d_0,d_1),(d_2,d_3)$ as $d_0,d_2$, and $d_3$ prefer the latter to the former.
  This proves part~(1) of the lemma.

  \smallskip
  
    Consider any node $v \ne z$ outside the gadget $D$. If $v$ is left unmatched in $M$ then we either have an alternating path $\rho_1 = (v,d_0)$-$(d_0,d_1)$-$(d_1,d_3)$
    or  an alternating path $\rho_2 = (v,d_0)$-$(d_0,d_2)$-$(d_2,d_1)$ with respect to $M$: the middle edge in each of these alternating paths belongs to $M$
    and the third edge is a {\em blocking edge} with respect to $M$. Both $\rho_1$ and $\rho_2$ are $M$-alternating paths in $G_M$ that start from an $M$-exposed
    node and end with a $(+,+)$ edge --- this is a forbidden structure for a popular matching (see Theorem~\ref{thr:characterize-popular}, path~(iii)).
    Hence every node $v \ne z$ in $H$ has to be matched in $M$. This proves part~(2). \qed
\end{proof}

Since the total number of nodes in $H$ is odd, at least 1 node has to be left unmatched in any matching in $H$.
Lemma~\ref{new-lemma1} shows that the node $z$ will be left unmatched in any popular matching in $H$.
For any popular matching $M$ in $H$,
the matching $M$ restricted to $G'$ (recall that $G'$ is $H \setminus D$) has to be popular on $G'$, otherwise it would
contradict the popularity of $M$ in $H$. We will now show the following converse of Lemma~\ref{new-lemma1}.

\begin{lemma}
    \label{new-lemma2}
    If $G'$ admits a popular matching that matches all its nodes except $z$ then $H$ admits a popular matching.
\end{lemma}
\begin{proof}
  Let $N$ be a popular matching in $G'$ that matches all its nodes except $z$. 
  Let $G'_N$ be the subgraph obtained by removing all edges labeled $(-,-)$ with respect to $N$ from $G'$.
  Since $N$ is popular in $G'$, it satisfies the 
  necessary and sufficient conditions for popularity given in Theorem \ref{thr:characterize-popular}.

  We claim $M = N \cup \{(d_0,d_1),(d_2,d_3)\}$ is a popular matching in $H$.
  We will now show that $M$ obeys those conditions in the subgraph $H_M$ obtained by deleting edges labeled $(-,-)$ with respect to $M$.
  The graph $H_M$ is the graph $G'_N$ along with some edges within the gadget $D$.

  There is no edge in $H_M$ between $D$ and any node in $G'$ since every edge in $H$ between  $D$ and a node in $G'$ is $(-,-)$.
  This is because for any such edge $(d_0,v)$, the node $d_0$ prefers $d_1$ (its partner in $M$) to $v$ and similarly, $v$ prefers each
  of its neighbors in $G'$ to $d_0$. Since $v \ne z$, we know that $N$ (and thus $M$) matches $v$ to one of its neighbors in $G'$.

  It is easy to check that $\{(d_0,d_1),(d_2,d_3)\}$ satisfies the 3 conditions from Theorem \ref{thr:characterize-popular} in the subgraph of $D$
  obtained by pruning $(-,-)$ edges. 
  We know that $N$ satisfies the 3 conditions from Theorem \ref{thr:characterize-popular} in $G'_N$.
  Thus $M$ satisfies all these 3 conditions in $H_M$.
  Hence $M$ is popular in $H$. \qed
\end{proof}

  Thus we have shown that $H$ admits a popular matching if and only if $G'$ admits a matching that matches all nodes except $z$.
  Since the latter problem is NP-hard, so is the former problem. Thus we have shown the following result.

  \begin{theorem}
  \label{main-thm}
  Given a graph $G = (V,E)$ with strict preference lists, the problem of deciding if $G$ admits a popular matching or not is NP-hard.
\end{theorem}

  \section{Popular matchings of minimum cost in bounded treewidth graphs}\label{sec:treewidth}

 In this section, we show a polynomial time algorithm to compute a minimum cost popular matching in a roommates instance $G = (V,E)$ with arbitrary edge costs
 under the assumption that $G$ has bounded treewidth.

Bounded treewidth is a classical assumption that often turns intractable problems into tractable ones. A typical example is Maximum Independent (Stable) Set (\texttt{MIS}), for which a polynomial-time algorithm exists in bounded treewidth graphs \cite{Bod}. \texttt{MIS} enjoys two nice properties, often shared by problems for which the bounded treewidth approach is successful. The first is \emph{monotonicity}: if $S$ is
an independent set in a graph $G$ and $G'$ is a subgraph of $G$, then the solution induced by $S$ on $G'$  is also feasible. The
second is \emph{locality}: in order to check if $S$ is an independent set, it suffices to verify, for each node of $S$, if any node of its neighborhood also belongs to $S$. 

Interestingly, similar properties do not hold for popular matchings. Indeed, popularity is not a local condition, since it may depend on how nodes far away in the graph are matched. Moreover, if we take a graph $G$ and a popular matching $M$, the subset of $M$ contained in an induced subgraph of $G$ may not be popular. Examples with both those features can be easily constructed by building on the characterization of popular matchings given by Theorem \ref{thr:characterize-popular}.

Our technique to prove the tractability of the minimum cost popular matching problem in the bounded treewidth case is as follows. We assume wlog that all matchings have different costs. Given a vertex separator $S$ of $G$ and a connected component $X$, we define the family of \emph{$(S,X)$-locally popular} matchings. This contains all matchings that may potentially be extended to a popular matching in the whole graph by adding edges not incident to $X\cup S$ (a formal definition is given in Section~\ref{sec:wpm}). As $(S,X)$-locally popular matchings can be exponentially many even in graphs of bounded treewidth
, we cannot store all of them. Instead, we divide them in classes (which we call \emph{tipping points}), and show that, if a matching $M$ of a class can be completed to a $(S',X')$-locally popular matching for some sets with $X'\supseteq X$ and $X'\cup S' \supseteq X \cup S$ by adding some matching $M'$ not incident to $X\cup S$, then all matchings of the same class can be extended via the same matching $M'$ to a $(S',X')$-locally popular matching. Hence, it is enough to keep only a representative for each class --- the one of minimum cost, which we call the \emph{leader}. Finally, we show how to iteratively construct tipping points and their leaders by building on a tree decomposition of the graph (see Algorithm \ref{Algo}). In particular, if the graph has bounded treewidth, a tree decomposition of bounded width can be computed in linear time~\cite{Bo96}, there will be only polynomially many tipping points and leaders, and they can be built in polynomial time, see Theorem \ref{thr:main}.




\subsection{Additional definitions}

For $k \in \N$, we  write $[k]:=\{1,\dots,k\}$. Recall that our input graph is $G = (V,E)$.
For $S\subseteq V$, we denote by $G\setminus S$ the subgraph of $G$ induced by $V\setminus S$. If $G\setminus S$ is disconnected, $S$ is said to be a \emph{vertex separator}. Given a matching $M$ on $G$ and a set $U\subseteq V$, the matching \emph{induced} by $M$ on $U$ is given by $E[U]\cap M$, and it is denoted by $M[U]$. For $u \in V$, $M(u)$ is the unique $v \in V$ such that $(u,v) \in M$ if such a $v$ exists, and $\emptyset$ otherwise. 

Now also fix $U\subseteq V$. Given a matching $M$ of $G$, we say that $M$ is a \emph{$U$-matching} if for all edges $(u,v) \in M$, we have $\{u,v\}\cap U \neq \emptyset$. Given a matching $M$ of $G$, the \emph{$U$-matching induced by $M$} is the set $M_U=\{(u,v) \in M : \{u,v\}\cap U \neq \emptyset\}$. 

Since all graphs we deal with are simple, throughout this section we represent paths and cycles as ordered set of nodes, with the first node of a cycle coinciding with the last. This also allows us to distinguish between the \emph{first} and \emph{last} node of a path. We will say that \emph{$e$ is an edge of $P$} if it is an edge between two consecutive nodes of $P$. A $U$-path $P$ is a path in $G$ where either the first or the last node of $P$ (or possibly both)
belongs to $U$, and all other vertices of $P$ do not lie in $U$ (i.e., if $P=(v_1,\dots,v_k)$, then $\emptyset \neq P\cap U \subseteq \{v_1,v_k\}$).


Two paths $P=(v_1,\dots,v_k)$, $P'=(v_1',\dots,v_{k'}')$ are called \emph{$U$-disjoint} if $P\cap P'\cap U=\emptyset$. $P$, $P'$ are said to be \emph{internally vertex-disjoint} if, for all $i \in [k]$ and $j \in [k']$, $v_i\neq v_j'$ with possibly the exception of $(i,j)\in \{(1,1),(1,k'),(k,1),(k,k')\}$. 
Consider paths $P_1,P_2,\dots,P_q$ of $G$, $q\geq 2$, whose union contains at least two distinct vertices and with the following properties:
\begin{itemize}
	\item $P_i \cap P_j= \emptyset$, unless $|i-j| \leq 1$ or $i,j \in \{1,q\}$;
	\item If $q\geq 3$, for $i=1,\dots,q-1$, the last node of $P_i$ is the first node of $P_{i+1}$, and $P_i, P_{i+1}$ are otherwise disjoint. Moreover, $P_1$ and $P_q$ are disjoint, with possibly the exception of the last node of $P_q$ coinciding with the first node of $P_1$;
	\item If $q=2$, the last node of $P_1$ coincides with the first node of $P_2$, the last node of $P_2$ may coincide with the first node of $P_1$, and paths $P_1$ and $P_2$ are otherwise disjoint.
\end{itemize}
The \emph{juxtaposition} of $P_1,\dots,P_q$ is the path (if the first node of $P_1$ and the last of $P_q$ are different) or cycle (otherwise) defined from $(P_1,P_2,\dots, P_q)$ as above by removing consecutive repeated vertices.

\smallskip

We assume that no two matchings of $G$ have the same cost. This can be achieved efficiently by standard perturbation techniques.



\subsection{Locally popular matchings, Configurations, and Tipping points}\label{sec:wpm}

In this section, we fix a graph $G = (V,E)$ together with a strict ranking of the neighbors of each node in $V$, a vertex separator $S$ of $G$, and a connected component $X$ of $G\setminus S$ (possibly $S=\emptyset$ and $X=V$).

Let $M$ be a $(X\cup S)$-matching of $G$. We will extensively work with graph $G_M[X \cup S]$ -- that is, the subgraph of $G_M$ induced by $X\cup S$. For a node $v$ in $S$ that is matched in $M$ to a neighbor outside $X \cup S$, the labels of edges of $G_M[X\cup S]$ incident to $v$ are a function of the edge $(v,M(v))$, however note that the edge $(v,M(v))$ is not in $G[X \cup S]$.

We say that $M$ is \emph{$(S,X)$-locally popular} if none of the structures (i), (ii), and (iii) from Theorem~\ref{thr:characterize-popular} is a subgraph of $G_M[X \cup S]$. We remark that a node that is not matched in $G_M[X\cup S]$ may still be $M$-covered (hence not $M$-exposed). If this happens, such a node cannot be the $M$-exposed node in the path (iii) from Theorem \ref{thr:characterize-popular}. The following simple proposition relates the definitions of popularity and $(S,X)$-local popularity.

\begin{prop}\label{prop:popular-will-be-weak}
	Let $G,S,X$ be as above, $S'$ be a vertex separator of $G$, and $X'$ be one of the connected components of $G\setminus S'$ such that $X' \cup S' \supseteq X \cup S$. Then:
	\begin{enumerate}
		\item  Let $M$ be a matching in $G$. Then $M$ is $(\emptyset,V)$-locally popular if and only if it is popular. 
		\item Let $M'$ be a $(S',X')$-locally popular matching. Let $M:=M'_{X \cup S}$. Then $M$ is an $(S,X)$-locally popular matching. 
	\end{enumerate}
\end{prop}

\begin{proof}
	1. It follows by definition and from the fact that $G_M[\emptyset\cup V]=G_M$.
	
	2. Suppose $M$ is not $(S,X)$-locally popular. Then, one of the structures from Theorem \ref{thr:characterize-popular} is a subgraph of $G_M[X \cup S]$ --- call it $P$. We claim that $P$ is a forbidden structure in $G_{M'}[X' \cup S']$ (again, in the sense of Theorem \ref{thr:characterize-popular}), hence $M'$ is not $(S',X')$-locally popular, a contradiction.

        Indeed, $X\cup S \subseteq X' \cup S'$, and none of the edges of $M'\setminus M$ is incident to $X\cup S$ by definition. We deduce $G_M[X\cup S]=G_{M'}[X \cup S]$ and a node of $X\cup S$ is $M$-exposed if and only if it is $M'$-exposed. This concludes the proof. \qed
\end{proof}

We now introduce the concepts of \emph{configuration} and \emph{tipping point} in order to partition the family of $(S,X)$-locally popular matchings into a (small) number of classes. These definitions are related to the existence of certain structures that are not forbidden in a popular matching, but restrict the capability of extending locally popular matchings to popular matchings of the whole graph. 

\smallskip

From now on, we also fix $M$ to be a $(X \cup S)$-matching of $G$.
Let ${\cal P}_M$ be the set of $M$-alternating paths of $G_M[X \cup S]$. We associate to each $P\in {\cal P}_M$ a $2$-dimensional \emph{parity} vector $\pi$, where the first component of $\pi$ is defined to be $0$ if the first edge of $P$ is a matching edge, $1$ otherwise. Similarly, the second component of the parity vector $\pi$ takes values $0$ or $1$, depending on whether the last edge of the path is a matching edge or not. We also associate to $P$ a two-dimensional \emph{level} vector $\ell=(e,p)$, where $e\in\{0,1,2\}$ is the number of $M$-exposed nodes of $P$, and $p \in \{0,1\}$ denotes the number of $(+,+)$ edges of $P$. If $p\geq 2$, then we say that $P$ is of level $\infty$. Note that not all parity-level combinations are possible. Moreover, if $P$ is of level $\infty$, then $M$ is not $(S,X)$-locally popular. Let
\begin{equation}
  \label{eq:U}
  U=((u_1,v_1),(u_2,v_2),\dots,(u_k,v_k)),
\end{equation}
for some $k \in \N$, where $u_1,v_1,u_2,\dots,v_k \in S \cup \{\emptyset\}$, under the additional condition that $u_i\neq v_i$ for all $i \in [k]$, all pairs are different, and each node of $S$ can appear at most twice in the collection. Moreover, let $L=(\ell_1,\dots, \ell_k)$, $\Pi=(\pi_1,\pi_2,\dots,\pi_k)$ and, for $i \in [k]$, $\ell_i \in \{0,1,2\} \times \{0,1\}$ and $\pi_i \in \{0,1\}\times \{0,1\}$. The triple ${\cal C}=(U,L,\Pi)$ is called an \emph{$(S,X)$-configuration}.

We say that $M$ is \emph{active at ${\cal C}$} if there exist pairwise $X$-disjoint $S$-paths $P^1,\dots,P^k \in {\cal P}_M$, such that, for $i \in [k]$, $P^i$: is of level $\ell_i$ and parity $\pi_i$; starts at $u_i \in S$ if $u_i\neq \emptyset$ and at some node of $X$ otherwise; ends at $v_i \in S$ if $v_i \neq \emptyset$ and at some node of $X$ otherwise. We call those paths the \emph{certificate of ${\cal C }$ at $M$}\footnote{Note that, if $M$ is active at ${\cal C}$, it is also active at the configurations e.g. obtained by permuting entries of $U$ (and of $L$ and $\Pi$ accordingly). This causes some redundancy, yet this will not affect our analysis, so we do not eliminate it.}.

Let $\{(U_i,L_i,\Pi_i)\}_{i=1,\dots,q}$ be the collection of all $(S,X)$-configurations at which $M$ is active. We call 
$$(\{(U_i,L_i,\Pi_i)\}_{i=1,\dots,q},M_S)$$ the \emph{$(S,X)$-tipping point of $M$}.

\begin{prop}
  \label{prop:poly-many-tipping}
	Let $G,X,S$ be as above. Let $M$ be an $(S,X)$-locally popular matching. The $(S,X)$-tipping point of $M$ is uniquely defined. On the other hand, there exists a function $g: \N \rightarrow \N$ such that the number of $(S,X)$-configurations is bounded by $g(|S|)$, and the collection of $(S,X)$-tipping points of $M$, where $M$ ranges over all $(S,X)$-locally popular matchings, has size at most $g(|S|)|V|^{|S|}$.
\end{prop}

\begin{proof}
	The first statement follows by definition. For the second and third: the number of $(S,X)$-configurations $(U,L,\Pi)$ is a function of the size of $S$ only, since all pairs from $U$ are different, and each node of $S$ can appear in at most two pairs from $U$. Moreover, the number of $S$-matchings of $G$ is upper bounded by $|V|^{|S|}$. \qed
\end{proof}

\begin{prop}\label{pro:match-same}
	Let $G,X,S$ be as above. Let $M$ be an $(S,X)$-locally popular matching and ${\cal T}$ be the $(S,X)$-tipping point of $M$. Let $N$ be another $(S,X)$-locally popular matching whose $(S,X)$-tipping point is also ${\cal T}$. Let $e\in \delta(v)$ for some $v \in S$. Then: 
	\begin{enumerate}
		\item $e \in M$ if and only if $e \in N$.
		\item Suppose $e \notin M$ and let $\star \in \{+,-\}$ be the label of $e$ at $v$ wrt $M$. Then $e \notin N$ and $\star$ is also the label of $e$ at $v$ wrt $N$.
	\end{enumerate}
\end{prop}

\begin{proof}
	Both statements immediately follow from $M_S = N_S$, which holds by definition of tipping point. \qed
\end{proof}

\subsection{Leaders}

Fix $G,S,X,M$ as in the previous section, and let ${\cal T}$ be the $(S,X)$-tipping point of $M$. $M$ is said to be the \emph{${\cal T}$-leader} if it is the one of minimum cost among all $(S,X)$-locally popular matchings whose $(S,X)$-tipping point is ${\cal T}$.
Due to the initial perturbation of costs, note that there is at most one ${\cal T}$-leader for each tipping point ${\cal T}$.



The following crucial lemma shows that, in order to find a min cost popular matching, it suffices to consider matchings that induce $(X\cup S)$-matchings that are ${\cal T}$-leaders, for any $(S,X)$-tipping point ${\cal T}$. The proof of Lemma~\ref{lem:only-leaders-left-alive} is given in Section~\ref{sec:lemma15-proof}.

\begin{lemma}\label{lem:only-leaders-left-alive}
	Suppose we are given $G,X,S,M$ as above. Let ${\cal T}$ be the $(S,X)$-tipping point of $M$, and assume that $M$ is not the ${\cal T}$-leader. Let $S'$ be a vertex separator of $G$, $X'$ a connected component of $G\setminus S'$ with the property that $X'\supseteq X$ and $S' \cup X' \supseteq S \cup X$ (possibly $S'=\emptyset$ and $X'=V$). Let $M'$ be a 
	$(S',X')$-locally popular matching such that $M'_{X\cup S}=M$. Let ${\cal T}'$ be the $(S',X')$-tipping point of $M'$. Then $M'$ is not a ${\cal T}'$-leader.
\end{lemma}

\subsection{Tree decomposition}

Let $G = (V,E)$ be a graph. A \emph{tree decomposition} of $G$ is a pair $(T, {\cal B})$ where $T$ is a tree and ${\cal B}=\{B(i):i\in V(T)\}$ is a family of subsets of $V(G)$, called \emph{bags}, one for each vertex of $T$, satisfying the following:
\begin{enumerate}
	\item For each $(u,v) \in E$, there is at least one bag $B \in {\cal B}$ such that $(u,v)\in B$.
	\item If $i\neq j \neq k \in V(T)$ are such that $k$ is on the unique path from $i$ to $j$ in $T$ then $B(i)\cap B(j) \subseteq B(k)$.
\end{enumerate} 

Note that the second property implies that, for any vertex $u\in V(G)$, the bags which contain $u$ form a subtree of $T$. We will sometimes abuse notation and denote by $B$ both a vertex of $V(T)$ and the bag corresponding to it. 

\smallskip

The \emph{width} of a tree decomposition $(T, {\cal B})$ is $\max\{|B|-1:B\in {\cal B}\}$. The \emph{treewidth} of $G$ is the minimum integer $\omega$ such that there is a tree decomposition of $G$ of width $\omega$.

Let $G = (V,E)$ be a graph and $(T, {\cal B})$ be a tree decomposition of $G$ of width $\omega$. Wlog we can assume that, for each pair of bags $B,B'$ adjacent in $T$, $B\cap B'$ is a vertex separator of $G$. Form a directed rooted tree by picking an arbitrary vertex as the root and orienting the remaining edges towards the root, and call this a \emph{directed tree decomposition}. In the directed tree, each node $X$ other than the root node has exactly one \emph{successor} $S(B)$, i.e., there exists exactly one node $S(B)$ such that $(B,S(B))$ is an arc of the directed tree decomposition. If $B$ is the root, we set $S(B)=\emptyset$. We also say that $B$ is a \emph{predecessor} of $S(B)$ and notice that a bag may have multiple predecessors, and has none if and only if it is a leaf of $T$.

\smallskip

\noindent{\em Dichotomic tree decomposition.}
Let $G$ be a graph and $(T,{\cal B})$ a directed tree decomposition of $G$ of width $\omega$. We can transform the directed tree decomposition of $G$ it into a directed tree decomposition of the same graph and width where every node has indegree at most $2$. We call such a tree decomposition \emph{dichotomic}. Indeed, suppose edges $(B_1,B), (B_2,B), \dots, (B_t,B)$ with $t\geq 3$ are part of the directed tree decomposition. Then we can create a copy $\bar B$ of $B$, add edge $(\bar B,B)$, and split the edges entering $B$ as follows: $(B_1,\bar B), \dots, (B_{\lceil t/2\rceil},\bar B)$ and $(B_{\lceil t/2\rceil+1},B), \dots, (B_{t},B)$. The indegree of $X$ and $\bar X$ is at most $\lfloor t/2\rfloor+1<t$. Notice that the digraph we obtain is still a tree decomposition of $G$, since each edge of $G$ is still covered by a bag, and the bags which contain any vertex $v\in G$ still form a continuous subtree.

If we repeat the above operation for a non-dichotomic tree once for each original node of indegree at least $3$, the maximum indegree over all nodes of the tree goes from $t>2$ to $\lfloor t/2\rfloor+1$, while the number of nodes is at most doubled.  Hence, we can iterate the operation so as to obtain a tree decomposition with at most a quadratic number of vertices, all of which have indegree at most $2$.

From here on, we assume without loss of generality that our directed tree decomposition is dichotomic. Let $T'$ be a subtree of $T$, and $V(T')$ the set of nodes contained in at least a bag of $T'$. We say that $T'$ is a \emph{closed subtree} of $T$ if $B \in V(T')$ and $B'$ is a predecessor of $B$, then $B'$ belongs to $T'$. By connectivity, for each node $B$ of a closed subtree $T'$, $S(B)$ also belongs to $T'$, with the exception of at most one node that we call the \emph{head} of $T'$ and denoted by $H(T')$. If $T'\neq T$, the successor of $H(T')$ exists and it is also called the \emph{successor of $T'$} and denoted by $S(T')$. Note that each bag $B$ is the head of exactly one closed subtree of $T$, that we denote by $T_B$.



\begin{remark}\label{rem:tree-decomposition}
	Let $(T,{\cal B})$ be a dichotomic directed tree decomposition. Then the following holds:
	let $T'$ be a closed subtree of $T$ and $B$ be the head of $T'$. The removal of $B$ partitions $T'\setminus B$ in at most $2$ closed subtrees. The successor of each of those subtrees is $B$.
\end{remark}

\subsection{The algorithm}

We now give our algorithm for computing a minimum weight popular matching, see Algorithm~\ref{Algo}. Note that Algorithm~\ref{Algo} relies on the subroutine \texttt{Update} described in Algorithm~\ref{algo:update}, and the implementations of some other subroutines are not completely defined. We give a formal description of those together with a complexity analysis in the proof of Theorem \ref{thr:main}.

Algorithm~\ref{Algo} takes as input a graph $G$ with strict preference lists as usual, and a dichotomic directed tree decomposition $(T,{\cal B})$ of $G$. It iteratively constructs the sequence of closed subtrees $T_B$ of $T$, with the first $T_B$ corresponding to a leaf $B$ of $T$, and the last to the root. For each $T_B$, let $S=B \cap S(B)$ and $X=V(T_B)\setminus S$. The algorithm constructs and stores in ${\cal L}_B$ all the pairs $(M,{\cal T})$, where  ${\cal T}$ is an $(S,X)$-tipping point and $M$ is the ${\cal T}$-leader. This set can be found by building on the corresponding sets for the (at most two) predecessors of $B$. Finally, of all matchings $M$ such that $(M,{\cal T}) \in {\cal L}_{B}$ -- with $B$ being the root -- the one of minimum cost is output. If at the end of any iteration when a bag $B$ is flagged, we have ${\cal L}_B=\emptyset$, we deduce that $G$ has no popular matching.

With a little abuse of notation, we will write $M \in {\cal L}_B$ if $(M,{\cal T}) \in {\cal L}_B$ for some ${\cal T}$. Similarly, we write ${\cal L}_B={\cal L}_B \cup \{M\}$ to mean that $(M,{\cal T})$ is added to ${\cal L}_B$ for an appropriate ${\cal T}$, and similarly for ${\cal L}_B \setminus \{M\}$.

\begin{algorithm}[h!]
	\caption{ }
	\label{Algo}
	\begin{algorithmic} [1]
		\REQUIRE A graph $G$, together with, for each node $v \in V$, a strict ranking of the neighbors of $v$. A dichotomic directed tree decomposition $(T,{\cal B})$ of $G$.
		
		\ENSURE A popular matching of minimum cost in $G$.
		\smallskip
		\STATE {For all $B \in {\cal B}$, label $B$ as \emph{unflagged}.}
		\STATE{Choose an unflagged bag $B$ whose predecessors are flagged, and flag $B$.}\label{st:pick-a-bag}
		\STATE{Let $T_B$ be the closed subtree of $T$ whose head is $B$, and set ${\cal L}_B=\emptyset$.}
		\STATE{Let $S=B\cap S(B)$, $X=V(T_B)\setminus S$.}
		\IF{$B$ has no predecessor in $T$}\label{step:if-leaf}
		\FOR{all $B$-matchings $M^*$ of $G$}
		\IF{$M^*$ is an $(S,X)$-locally popular matching}\label{step:if1}
		\STATE{\texttt{Update}($M^*$, ${\cal L}_B$)}
		\ENDIF
		\ENDFOR
		\ELSE
		\STATE{Let $S_1=B\cap B_1$ and (possibly) $S_2=B\cap B_2$,  where $B_1$ and (possibly) $B_2$ are the predecessors of $B$.}
		\FOR{all $B$-matchings $M$ of $G$, all $M_1\in {\cal L}_{B_1}$ and (possibly) $M_2 \in {\cal L}_{B_2}$}
		\IF{$(M_1)_{S_1}=M_{S_1}$ and (possibly) $(M_2)_{S_2}=M_{S_2}$}\label{st:same-induced}
		\STATE{Let $M^*=M \cup M_1 \cup M_2$}\label{st:M-star}
\IF{$M^*$ is an $(S,X)$-locally popular matching of $G$}
		\STATE{\texttt{Update}($M^*$, ${\cal L}_B$)}
		\ENDIF
		\ENDIF
		\ENDFOR
		\ENDIF
		\IF{${\cal L}_B=\emptyset$}
		\STATE{output: $G$ has no popular matching.} 
 		\ELSIF{there is a bag $B \in {\cal B}$ that is unflagged}
		\STATE{Go to Step \ref{st:pick-a-bag}.}
		\ENDIF 
		\STATE{Let $B$ be the head of $T$. Output the matching of minimum cost from ${\cal L}_B$.}

	\end{algorithmic}
\end{algorithm}

\begin{algorithm}[h]
	\begin{algorithmic}
		\caption{\texttt{Update}}\label{algo:update}
		\REQUIRE $M^*$, ${\cal L}_B$
		\STATE{Let ${\cal T}$ be the $(S,X)$-tipping point of $M^*$.}\label{st:find-T}
		\IF{there exists $M' \in {\cal L}_B$ whose tipping point is ${\cal T}$}
		\IF{$c(M^*)<c(M')$}
		\STATE{Set ${\cal L}_B={\cal L}_B \setminus \{M'\}\cup \{M^*\}$.}
		\ENDIF
		\ELSE	
		\STATE{Set ${\cal L}_B={\cal L}_B \cup \{M^*\}$.}
		\ENDIF
		
	\end{algorithmic}
\end{algorithm}

\subsection{Analysis}

The goal of this section is to prove the following result.

\begin{theorem}\label{thr:main}
	Algorithm \ref{Algo} is correct. If the treewidth of $G = (V,E)$ is upper bounded by a constant~$\omega$, then it can be implemented to run in time $O(|V|^{3\omega +7})$.
\end{theorem}

We assume throughout the proof that every bag $B$ that is not a leaf has two predecessors, as the (simpler) case where some $B$ has only one predecessor follows in a similar fashion. For a bag $B$ we write $S:=B\cap S(B)$ and $X:=V(T_B)\setminus S$; if $B$ is not a leaf, we denote by $B_1$ and $B_2$ its predecessors, and write $S_i:=B\cap B_i$, $X_i=V(T_{B_i})\setminus S_i$ for $i=1,2$. We start by proving the correctness of Algorithm \ref{Algo}. We show that, at the end of the iteration where bag $B$ is flagged,
\begin{itemize}
\item[$(*)$] ${\cal L}_B$ contains exactly all ${\cal T}$-leaders, for all $(S,X)$-tipping points ${\cal T}$ for which a ${\cal T}$-leader exists.
\end{itemize} Suppose $(*)$ is proved. Then, when $B$ is the root, ${\cal L}_B$ contains only popular matchings (by Proposition~\ref{prop:popular-will-be-weak}), and the one of minimum cost among those is the popular matching of minimum cost. Again by Proposition \ref{prop:popular-will-be-weak}, if ${\cal L}_B=\emptyset$ at the end of the iteration where $B$ is flagged, then $G$ has no popular matching, and the output of Algorithm \ref{Algo} is again correct. 

\smallskip

The proof of $(*)$ is by induction on the number of nodes $n_B$ of $T_B$. If $n_B=1$, then the condition from the {\bf if} statement in Step \ref{step:if-leaf} is verified. In this case, the statement is immediate, since we enumerate over all possible $B$-matchings of $G$, check those that are $(S,X)$-locally popular, and for each $(S,X)$-tipping point ${\cal T}$, keep the $(S,X)$-locally popular matching of minimum cost active at ${\cal T}$.

Now suppose $n_B > 1$. By induction hypothesis, for $i=1,2$, all matchings that are stored in ${\cal L}_{B_i}$ are exactly all ${\cal T}$-leaders, for all $(S_i,X_i)$-tipping points ${\cal T}$ for which a ${\cal T}$-leader exists (i.e., there is at least a $(S_i,X_i)$-locally popular matching active at ${\cal T}$). Now let ${\cal T}$ be an $(S,X)$-tipping point for which a ${\cal T}$-leader $\hat M$ exists. It suffices to show that one of the matchings $M^*$ constructed at Step \ref{st:M-star} is indeed $\hat M$.

Since $\hat M$ is $(S,X)$-locally popular and $X_i \cup S_i \subseteq X \cup S$, matching $M_i:=\hat M_{X_i \cup S_i}$ is also $(S_i,X_i)$-locally popular by Proposition \ref{prop:popular-will-be-weak}. By Lemma \ref{lem:only-leaders-left-alive}, $M_i$ is a $(S_i,X_i)$-leader. By induction, $M_i \in {\cal L}_{B_i}$. On the other hand, $M:=\hat M_{B}$ is a $B$-matching of $G$ such that $M_{S_i}=(M_{i})_{S_i}$ for $i=1,2$. Since we enumerate all $B$-matchings of $G$, as well as all matchings from ${\cal L}_{B_1}$ and ${\cal L}_{B_2}$, matching $\hat M$ is eventually enumerated.

\smallskip

\noindent{\em Running time analysis.}
We now bound the running time of Algorithm~\ref{Algo}. We will use the following general fact proved in the claim below: if we are given an $(X\cup S)$-matching $M$ in a graph $H$ and $|X\cup S|$ is  bounded, then one can check $(S,X)$-local popularity of $M$. If $M$ is $(S,X)$-locally popular, then we can efficiently find the $(S,X)$-tipping point at which $M$ is active. 

\begin{new-claim}\label{cl:everything-is-checkable}

	Let $H = (U,F)$ be a graph, $S$ a vertex separator of $H$, $X$ a connected component of $H\setminus S$. Let $M$ be a $(X \cup S)$-matching of $H$. Assume $|X\cup S|$ is upper bounded by a constant. Then in $O(|F|)$ time one can:
	\begin{enumerate}
		\item check if $M$ is $(S,X)$-locally popular in $H$ and, if it is,
		\item find the $(S,X)$-tipping point (in $H$) at which $M$ is active.
	\end{enumerate}
\end{new-claim}

\begin{proof}
	Building graph $H_M[X\cup S]$ takes $O(|F|)$ time. Since $|X\cup S|$ is upper bounded by a constant, it takes constant time to enumerate all paths and cycles of $H_M[X \cup S]$, and to check if any of these violates the definition of $(S,X)$-local popularity. In case $M$ is indeed $(S,X)$-locally popular, for any family of paths ${\cal P}$ in $H_M[X\cup S]$ and $(S,X)$-configurations ${\cal C}$, one can check in constant time if ${\cal P}$ is a certificate for $M$ at ${\cal C}$. Since by Proposition \ref{prop:poly-many-tipping} the number of $(S,X)$-configurations is upper bounded by $g(\omega+1)$, which is a constant by hypothesis, the claim follows. \hfill $\lozenge$
\end{proof}


Observe that $|B|\leq \omega +1$, so enumerating all $B$-matchings takes time $O(|V|^{\omega+1})$.
Because of Claim \ref{cl:everything-is-checkable} and Step \ref{step:if1} (in Algorithm~\ref{Algo}), finding the tipping point of any
$B$-matching $M^*$ can be performed in time $O(|E|)=O(|V|^2)$. Moreover, once the tipping point of ${M^*}$ is computed, the \texttt{Update} function can be implemented to run in time $O(|V|^{\omega+1})$, since ${\cal L}_B$ at each step will contain at most one matching per tipping point, and by Proposition \ref{prop:poly-many-tipping}, there are at most $g(\omega+1)|V|^{\omega +1}$ many tipping points.

 We conclude that, when $n_B=1$, the iteration when $B$ is flagged runs in  time $O(|V|^{2\omega +2})$. We now prove, by induction on $n_B$, that the iteration when any $B$ is flagged can also be implemented to run in time $O(|V|^{3\omega +5})$, the base case having been just proved. Multiplying this by $O(|V|^2)$ (the number of bags of the tree decomposition), we obtain the desired bound.

Assume $n_B>1$. Because of what was discussed above, given an $(S,X)$-configuration ${\cal C}$ and a matching $M^*=M\cup M_1 \cup M_2$,
where $M$ is a $B$-matching of $G$ and for $i=1,2$ $M_i \in {\cal L}_{B_i}$,
it is enough to give an upper bound on the time needed to decide if $M^*$ is $(S,X)$-locally popular and if it is active at ${\cal C}$.
Recall that the number of $B$-matchings is $O(|V|^{\omega+1})$, while $|{\cal L}_{B_i}|=O(|V|^{\omega+1})$ and the number of $(S,X)$-configuration
is at most $g(\omega + 1)$ by Proposition \ref{prop:poly-many-tipping}.

The condition from Step \ref{st:same-induced} can be verified in time $O(|E|)$. Hence, we assume that $(M_1)_{S_1}=M_{S_1}$ and $(M_2)_{S_2}=M_{S_2}$. We start with a simple claim.

\begin{new-claim}\label{cl:all-paths-are-ivd}
	For $i=1,2$, let ${\cal C}_i=(U_i,L_i,\Pi_i)$ be a configuration at which $M_i$ is active, and let ${\cal P}_i$ be a certificate of ${\cal C}_i$ at $M_i$. Then all paths from ${\cal P}_1,{\cal P}_2$ are pairwise internally vertex-disjoint. 
\end{new-claim}
\begin{proof}
	By definition of certificate, paths from ${\cal P}_1$ are pairwise internally vertex-disjoint, and similarly so are paths from ${\cal P}_2$. Now let $P_1 \in {\cal P}_1$, $P_2 \in {\cal P}_2$. For $i=1,2$, $P_i\subseteq V(T_{B_i})$ and $P_i$ is an $S_i$-path. Since $V(T_{B_1})\cap V(T_{B_2})\subseteq S_1 \cup S_2$, the claim follows. \hfill $\lozenge$
\end{proof}

Now fix an $(S_1,X_1)$-configuration ${\cal C}_1=(U_1,L_1,\Pi_1)$ in $G$ at which $M_1$ is active, and an $(S_2,X_2)$-configuration ${\cal C}_2=(U_2,L_2,\Pi_2)$ in $G$ at which $M_2$ is active. Consider the graph $H$ and matching $M_H$ obtained as follows. Start from $H=G_M[B]$ and the corresponding matching $M_H=M[B]$, and let $(u,v)$ be the first pair from $U_1$. Assume $u,v$ are matched by $M_H$, $\ell_1=(0,1)$, $\pi_1=(0,0)$, the other cases following in a similar fashion. 
Add to $H$ and $M_H$ new nodes $u',w',z',v'$, matching edges $(u',w')$, $(z',v')$, $(+,+)$ edge $(w',z')$, and $(+,-)$ edges $(u,u')$ and $(v,v')$. Note that $(u,u',w',z',v',v)$ is a $M_H$-alternating $S_1$-path of parity $\pi_1$ and level $\ell_1$ starting at $u$ and ending at $v$. We call it the \emph{shortcut} of the path with endpoints $u,v$. Repeat this for all pairs from $U_1$ and $U_2$, adding at each time new nodes and an appropriate path. Note that this adds to $H$ a bounded number of nodes. This means that the graph $H$ and matching $M_H$ can be constructed in time $O(|E|)$.

\begin{new-claim}\label{cl:almost-there}
	Let $M^*$ be an $(S\cup X)$-matching. $M^*$ is not an $(S,X)$-locally popular matching in $G$ if and only if there exist an $(S_i,X_i)$-configuration ${\cal C}_i$ at which $M_i$ is active for $i=1,2$, such that, if we construct graph $H$ and matching $M_H$ as above, then $M_H$ is not $(S,V(H)\setminus S)$-locally popular in $H$.   
\end{new-claim}

\begin{proof}
	Consider the forbidden subgraph $P$ (from Theorem \ref{thr:characterize-popular}) of $G_{M^*}[X \cup S]$, whose existence certifies that $M^*$ is not $(S,X)$-locally popular in $G$, and take $P$ that is minimal with this property. Using an immediate modification to Algorithm \ref{algo:decompose} (given in Section~\ref{sec:lemma15-proof}), we can write $P$ as the juxtaposition of $(P_1,P_2,\dots,P_k)$ where each $P_j$ is: either an $S_i$-path contained in $G[X_i \cup S_i]$ for $i\in \{1,2\}$, or a path contained in $G[B]$. Collect all such $S_1$-paths from $(P_1,\dots,P_k)$ not contained in $G[B]$ in ${\cal P}_1$, and all $S_2$-paths from $(P_1,\dots,P_k)$ not contained in $G[B]$ in ${\cal P}_2$. One easily checks that ${\cal P}_i$ form the certificate of a certain $(S_i,X_i)$-configuration ${\cal C}_i$ at $M_i$.
	
	Now consider the graph $H$ and matching $M_H$ obtained from configurations ${\cal C}_1$ and ${\cal C}_2$. By replacing each path from ${\cal P}_1$ and ${\cal P}_2$ with the corresponding shortcut, we deduce that $M_H$ is not an $(S,V(H)\setminus S)$-locally popular matching in $H$. 
	
	The opposite direction follows in a similar (inverse) fashion: start from a forbidden path in $H$ that certifies that $M_H$ is not $(S,V(H)\setminus S)$-locally popular in $H$, write it as $(P_1,\dots,...,P_k)$, replace each of those subpaths with the appropriate certificate paths so as to obtain an $M^*$-alternating forbidden path by Claim \ref{cl:all-paths-are-ivd}.\hfill $\lozenge$
\end{proof}

The proof of the following claim follows in a similar fashion to the proof of Claim \ref{cl:almost-there}.

\begin{new-claim}\label{cl:there,there}
	Let $M^*$ be an $(S,X)$-locally popular matching in $G$ and let ${\cal C}$ be an $(S,X)$-configuration. Then $M^*$ is active at ${\cal C}$ if and only if there exist a $(S_i,X_i)$-configuration ${\cal C}_i$ at which $M_i$ is active for $i=1,2$ such that if we construct graph $H$ and matching $M_H$, then $M_H$ is active at ${\cal C}$ (here we interpret ${\cal C}$ as an $(S,V(H)\setminus S)$-configuration in $H$, which is allowed since $S\subseteq V(H)$).
\end{new-claim}

Because of Claim \ref{cl:almost-there}, we can check if $M^*$ is an $(S,X)$-locally popular matching in $G$ by checking, for each pair of $(S_i,X_i)$-configurations ${\cal C}_i$ at which $M_i$ is active for $i=1,2$, if the corresponding $M_H$ is $(S,V(H)\setminus S)$-locally popular in $H$. For $i=1,2$, the number of $(S_i,X_i)$-configurations is at most $g(\omega+1)$. Since the number of nodes of $H$ is bounded (we start with a bounded set of nodes and we add a bounded number of new nodes), testing if $M_H$ is $(S,V(H)\setminus S)$-locally popular in $H$ can be done in time $O(|E|)=O(|V|^2)$ by Claim \ref{cl:everything-is-checkable}. Similarly, if $M^*$ is $(S,X)$-locally popular in $G$, we can find its tipping point in time $O(|V|^2)$ by repeatedly applying Claims~\ref{cl:there,there} and \ref{cl:everything-is-checkable}. Hence, the iteration when $B$ is flagged can be performed in time $O(|V|^{3\omega +5})$. This concludes the proof of Theorem~\ref{thr:main}.

\subsection{Proof of Lemma~\ref{lem:only-leaders-left-alive}}
\label{sec:lemma15-proof}

Let $N$ be the $(S,X)$-locally popular matching that is the ${\cal T}$-leader. This implies that $M_S=N_S$, and $N$ is a $(X\cup S)$-matching. Define $N':=N \cup (M'\setminus M)$, and notice this is a disjoint union. Hence, $N'_{X \cup S}=N$. We will now show that $N'$ is a $(S',X')$-locally popular matching whose $(S'\cup X')$-tipping point is ${\cal T}'$. We then have:
	$$c(N')  = c(N) + c(N'\setminus N) < c(M) + c(M'\setminus M) = c(M'),$$ concluding the proof. We start with some claims, the first two of which immediately follow from construction.
	
	\begin{new-claim}\label{cl:sameG}
		$G_{M'}[X \cup S]=G_M[X\cup S]$ and  $G_{N'}[X\cup S]=G_{N}[X \cup S]$. 
	\end{new-claim}
	
	\begin{new-claim}\label{cl:reverseM'}
		$M'=M\cup(N'\setminus N)$, and the latter is a disjoint union. 
	\end{new-claim}

	\begin{new-claim}\label{cl:one}
		Let $e \in E$ not be incident to a node of $X$. 
		\begin{enumerate}
			\item[a)]\label{first} $e \in N'$ if and only if $e \in M'$. 
			\item[b)]\label{second} Let $e \notin N'$. Then $e \in E(G_{N'})$ if and only if $ e \in E(G_{M'})$. Moreover, if $e \in E(G_{N'})$, then it has the same labels in $G_{N'}$ and in $G_{M'}$. 
		\end{enumerate}
	\end{new-claim}

	\begin{proof}
		a) Suppose $e \in N'$. If $e \in N$, then $e \in N_S = M_S \subseteq M \subseteq M'$, where equality holds by hypothesis and inclusions by definition. Else, $e \in (N'\setminus N)\subseteq M'$. The argument can be reversed to show the opposite direction.
		
		\smallskip
		
		b) Suppose $e=(u,v) \notin N'$. By part a), $e \notin M'$. If $u \notin S$, then $N'(u)=M'(u)$ by part a). We deduce that the label of $e$ at $u$ coincides in $M'$ and $N'$. If instead $u \in S$, then the label of $e$ at $u$ in $G_{M}$ and $G_{N}$ coincide by Proposition \ref{pro:match-same}. The claim follows since $M'_S=M_S=N_S=N'_S$. \hfill $\lozenge$ \end{proof}

	\begin{new-claim}\label{cl:matching,indeed}
		$N'$ is an $(X' \cup S')$-matching of $G$.
	\end{new-claim}
	
	\begin{proof}
		Let us first show that it is a matching of $G$. Recall that $M'_{X \cup S}=M$. Hence, $M'\setminus M$ has no edge incident to $X \cup S$. Since by hypothesis $N$ is an $(X \cup S)$-matching, $N$ has no edge incident to nodes other than $X \cup S$. Hence $N'$ is the union of two node-disjoint matchings of $G$, hence a matching. 
		
		Now by hypothesis, $S \cup X \subseteq S' \cup X'$. Hence, $N$ has no edge in $G\setminus (S' \cup X')$. Moreover, $M'$ is by hypothesis an $(S',X')$-locally popular matching, so it is an $(S'\cup X')$-matching. Hence, $N'$ is an $(S'\cup X')$-matching. \hfill $\lozenge$\end{proof}

	\begin{new-claim}\label{cl:fromPtoPhat}
		Let $k \in \N$ and  $P^1 ,\dots, P^k \in {\cal P}_{M'}$ be pairwise internally  vertex-disjoint and $X$-disjoint. Assume that, for $j \in [k]$, path $P^j$ is of level $\ell_j \in \{0,1,2\} \times \{0,1\}$ and parity $\Pi_j \in\{0,1\}\times\{0,1\}$, that every node appears at most twice in $P^1\dots,P^k$, no pair of paths have exactly the same endpoints, and no path is contained in $X$. 
		Then there exist pairwise internally vertex-disjoint and $X$-disjoint paths  $Q^1,\dots,Q^k \in {\cal P}_{N'}$ with the following properties:
		\begin{enumerate}
			\item For $j \in [k]$, $Q^j$ is of level $\ell_j$ and parity $\Pi_j$, and $P^j \cap S'=Q^j \cap S'$. 
			\item For $j \in [k]$, let $u_j$ (resp. $v_j$) be the first (resp. last) node of $P^j$. If $u_j$ (resp. $v_j$) $\notin X$, then $u_j$ (resp. $v_j$) is the first (resp. last) node of $Q^j$. If $u_j$ (resp. $v_j$) $\in X$, the first (resp. last) node of $Q^j$ also belongs to $X$.
		\end{enumerate}
		Moreover,
		\begin{enumerate}
			\item[3.]
			1,2 also hold if we switch the roles of $M'$ and $N'$.
		\end{enumerate}
	\end{new-claim}
	
	\begin{proof}
		The proof of the claim will only assume that $M$ and $N$ are $(S,X)$-locally popular matchings with the same $(S,X)$-tipping point. Recall that $N'=N\cup (M'\setminus M)$, $M'=M\cup(N'\setminus N)$ (by Claim~\ref{cl:reverseM'}), and that those unions are disjoint. Therefore, the roles of $M$ and $N$ (hence those of $M'$ and $N'$) can be exchanged and the conclusions preserved. This proves 3 (assuming 1,2).
		
		\smallskip

		
		Let us consider any path $P$ from $P^1,\dots,P^k$, where we omit the superscript for the sake of readability. If $P \cap X=\emptyset$, we write $P_\infty:=P$. Else, we write $P$ as the juxtaposition of $P_0,P_1,P_2,\dots, P_q,P_\infty$, defined through the procedure described in Algorithm \ref{algo:decompose}. Note that some of those paths  may consist of a single node. We call them \emph{trivial}. Note that a trivial path can be removed from the juxtaposition without changing the resulting $P$. Hence, in the following, we assume that all paths are non-trivial (without changing the subscripts of the remaining non-trivial paths), and we let $q$ be the subscript of the last non-trivial path (other than $\infty$).
		
		\begin{algorithm}[h!]
			\caption{ }
			\label{algo:decompose}
			\begin{algorithmic} [1]
				
				\STATE {Let $v$ be the first vertex of $P$.}
				\IF{$v \in X$}\label{step:if-X}
				\STATE{Let $P_0=\{v\}$ and $u=v$.}
				\ELSE
				\STATE{Let $u'$ be the first vertex of $P$ that is contained in $X$, and $u$ the predecessor of $u'$ in $P$. Note that $u \in S$. Let $P_0$ be the subpath of $P$ between $v$ and $u$ (possibly $u=v$ and $P_0=\{u\}$).}
				\ENDIF
				\STATE{Let $i=1$.}
				\STATE{Starting from $u$, traverse $P$ until the first node $v \in S$, $v\neq u$ is encountered.}\label{step:iterative}
				\IF{such $v$ does not exist}
				\STATE{Set $q=i-1$ and go to Step \ref{step:output}.}
				\ELSE{}
				\STATE{Let $P_i$ be the subpath between $u$ and $v$. Set $i=i+1$ and $u=v$.}
				\STATE{Starting from $u$, traverse $P$ until the first node $u' \in X$ is encountered.}
				\IF{such $u'$ does not exist}
				\STATE{Set $q=i-1$ and go to Step \ref{step:output}.}
				\ELSE{}	\STATE{Let $P_i$ be the subpath of $P$ between $u$ and the predecessor $v$ of $u'$ in $P$. Note that $v \in S$ and possibly $v=u$.}
				\STATE{Set $u=v$, $i=i+1$, and go to Step \ref{step:iterative}.}
				\ENDIF
				\ENDIF
				\STATE{Let $P_\infty$ be the subpath of $P$ between $u$ and the last node of $P$.}\label{step:output}
				\STATE{Output $P_0,P_1,\dots,P_{q},P_\infty$.}
			\end{algorithmic}
		\end{algorithm}
		
		The following fact immediately holds by construction.
		
		\begin{fact}\label{fact:basic}
			Let $\{P_i\}_{i\in [q]\cup\{0,\infty\}}$ be the output of Algorithm \ref{algo:decompose} on input $P=P^j$ for some $j \in [k]$. Then: (i) $P$ is the juxtaposition of $P_0,P_1,\dots,P_{q},P_\infty$. (ii) For all $i \in [q] \cup \{0,\infty\}$, $P_i \in {\cal P}_{M'}$. (iii) For $i\in \N$, $i \geq 2$, the first and last nodes of $P_i$ belong to $S$. (iv) For $i \in \N$, $i$ odd, $P_i$ is an $S$-path and $P_i\subseteq X \cup S$. (v) For $i$ even, $P_i \cap X=\emptyset$. (vi) Assume $q\geq 1$. Then $P_1$ ends at a node of $S$, and starts either at a node of $X$ (if the {\bf if} condition in Step \ref{step:if-X} of Algorithm \ref{algo:decompose} is satisfied), or at a node of $S$ (otherwise). 
			(vii) If $P_\infty \cap X \neq \emptyset$, then $P_\infty$ is an $S$-path of $G_{M'}[X\cup S]$ ending at some node of $X$. 
		\end{fact}
		Let us call a subpath $P_i$ of $P$ \emph{hidden} if $i \equiv 1 \bmod 2$ or $i=\infty$ and $P_\infty \cap X\neq \emptyset$ (i.e., condition (vii) above is verified). Now consider the collection of paths $\{P_i^j\}$ defined above, for $j \in [k]$\footnote{We introduced the superscripts back, to distinguish the paths produced by applying Algorithm \ref{algo:decompose} to $P^1,\dots,P^k$.}.
		\begin{fact}\label{fact:paths}
			Consider the family of paths $\{P_i^j\}$ with $j \in [k]$. Then:
			\begin{enumerate}
				\item $\{P_i^j\}$ is a collection of pairwise internally vertex-disjoint paths and each node appears at most twice in the collection.
				\item The restriction of the collection $\{P_i^j\}$ to hidden paths is a family of pairwise $X$-disjoint $S$-paths from ${\cal P}_M$, and the level of each of those path is not $\infty$.
				\item If $P_i^j \cap P_{i'}^{j'} \neq \emptyset$, then either (a) $j=j'$ and ($|i-i'|\leq 1 $ or $i,i' \in\{q_j,\infty\}$), or (b) $i$ (resp. $i'$) is the first or last path of the juxtaposition leading to $P^j$ (resp. $P^{j'}$). 
			\end{enumerate}
		\end{fact}
		
		\begin{proof}
			Part 1 follows from hypothesis, the fact that paths $P^j$, $j \in [k]$ form a collection of pairwise internally vertex-disjoint paths and Fact \ref{fact:basic}.
			
			We now prove part 2. Restrict the collection $\{P_i^j\}$ to hidden paths. We already argued in Fact \ref{fact:basic} that those are $S$-paths. As they are subpaths of $M'$-alternating paths, they are also $M'$-alternating. Using again Fact \ref{fact:basic}, they are contained in $X \cup S$. Using Claim \ref{cl:sameG}, we deduce they are paths from ${\cal P}_{M}$. The fact that they are $X$-disjoint follows by hypothesis for $j\neq j'$, and by Fact \ref{fact:basic} for $j=j'$. The fact that $M'$ is $(S',X')$-locally popular implies that they cannot be of level $\infty$.  This proves part 2. 
			
						We now prove part 3. Let $P_i^j \cap P_{i'}^{j'} \neq \emptyset$. If $j=j'$, by Fact \ref{fact:basic} the only possibility is that $P^j_i$, $P^{j}_{i'}$ are consecutive paths in the juxtaposition leading to $P$. This is case (a) in 3. Else, $j\neq j'$ and case~(b) follows from the fact that $P^j$ and $P^{j'}$ are internally vertex-disjoint, concluding the proof of 3.
		\hfill $\lozenge$\end{proof}
		Consider the collection of hidden paths $\{P_i^j\}$ for $j \in [k]$, and let $u_i^j$ (resp. $v_i^j$) be the first (resp. last) node of each of those paths. Consider ${\cal C}=(U,L,\Pi)$ defined as follows: $U$ is the ordered collection of pairs:
		\begin{itemize}
			\item $(u^j_1,v^j_1)$ if the first node of $P_1^j$ belongs to $S$, and $(\emptyset,v^j_1)$ otherwise.
			\item  $(u^j_i,v^j_i)$ for $i\geq 3$ odd and
			\item  $(u^j_\infty,\emptyset)$ if $P^j_\infty$ is hidden.
		\end{itemize} 
		$L$ is the ordered collection of levels $\ell_i^j$ of the hidden paths $P_i ^j$, while $\Pi$ is the ordered collection of their parities. Using Fact \ref{fact:basic}, Fact \ref{fact:paths} and the hypothesis, one easily verifies the following. 

\begin{fact}
${\cal C}$ is an $(S,X)$-configuration. $M$ is active at ${\cal C}$, and the collection of hidden paths $\{P^j_i\}$ for $j \in [k]$ is a certificate of ${\cal C}$ at $M$.
\end{fact}

		
		
		By hypothesis, $N$ is also active at ${\cal C}$. Hence, for all pairs $(i,j)$ such that $P_i^j$ is hidden, we can find pairwise $X$-disjoint $S$-paths $Q_i^j \in {\cal P}_{N}$, with each $P_i^j$ starting (resp. ending) at node $u_i^j$ (resp. $v_i^j$) if this belongs to $S$, and at a node of $X$ otherwise, of the appropriate level and parity.
		
		For $j \in [k]$, recall that $P^j$ is the juxtaposition of $(P^j_0,P^j_1,\dots,P_{q_j}^j,P^j_\infty)$, from which we removed the trivial paths. Consider the path $Q^j$ obtained by replacing, in the juxtaposition above, each hidden path $P^j_i$ with the corresponding $Q^j_i$. We conclude the proof of the claim by showing that the collection of $Q^j$, $j \in [k]$, satisfies the claim. 
		
		Fix $j \in [k]$. Let us first argue that $Q^j$ is a path in $G$ that satisfies property 2 from the statement of the claim. Note that, by construction, for all hidden $P^j_i$, the first (resp. last) vertex of $P^j_i$ coincides with the first (resp. last) vertex of $Q_i^j$, with possibly the exception of the first vertex of $P_0^j$ (resp. last vertex of $P_\infty^j$) if it belongs to $X$.  Moreover, since all $Q^j_i$ are pairwise $X$-disjoint $S$-paths with no edge in $G\setminus X$, each $Q^j$ is a path in $G$. The statement follows.
		
		Let us now argue that $Q^j$ is $N'$-alternating. Since $N'[(X'\cup S') \setminus X]=M'[(X'\cup S') \setminus X]$ and all $P^j_i$ that are not hidden are $M'$-alternating and do not intersect $X$, we conclude that all such $P^j_i$ are also $N'$-alternating. For each hidden $P^j_i$, $Q^j_i$ is an ${N}$-alternating path in $G_{N}[X \cup S]$ by construction. Since $G_{N}[X \cup S]=G_{N'}[X \cup S]$ by Claim \ref{cl:sameG}, each $Q^i_j$ is also $N'$-alternating. 
		By construction, the parity of each hidden $P_i^j$ coincides with that of the corresponding $Q_i^j$.  We conclude that each $Q^j$ is $N'$-alternating. 
		
		Let us now show that $Q^j \in {\cal P}_{N'}$, and it is of the same level and parity of $P^j$. Let $P^j_i$ be non-hidden. By Fact \ref{fact:basic}, it is contained in $G[X'\cup S']\setminus X$. By Claim \ref{cl:one}, $P^j_i\cap N'=P^j_i\cap M'$, and each edge from $P^j_i\setminus N'$ has the same labels in $G_{M'}[X' \cup S']$ and $G_{N'}[X' \cup S']$. Now consider a hidden path $P^j_i$. By construction, $Q^j_i$ is a path of $G_{N}[X\cup S]=G_{N'}[X \cup S]$ of the same parity and level of $P^j_i$. The statement follows.
		
		Now observe that property 1 of the claim holds, since the hypothesis $X\subseteq X'$ implies $S' \cap X=\emptyset$, hence we have: $$Q^j \cap S'= (Q^j \cap (S' \cap S))\cup (Q^j \cap (S'\setminus S))= (P^j \cap (S' \cap S))\cup (P^j \cap (S'\setminus S)) = P^j \cap S'.$$ 
		Finally, let us observe that $Q^1,\dots, Q^k$ are internally vertex-disjoint and $X$-disjoint. The $X$-disjointness follows from the fact that all $Q^j_i$ corresponding to hidden $P^j_i$ are $X$-disjoint, while other $P^j_i$ do not intersect $X$ by Fact \ref{fact:basic}. The fact that $Q^1,\dots, Q^k$ are pairwise internally vertex-disjoint follows from Fact \ref{fact:paths}. \hfill $\lozenge$ \end{proof}
	
	In order to conclude the proof of the lemma, we need to prove that $N'$ is a $(S',X')$-locally popular matching whose tipping point is ${\cal T}'$.

\smallskip
        
\noindent{\em $(S',X')$-local popularity.} Let us first show that $N'$ is $(S',X')$-locally popular. Suppose it is not, then there exists a forbidden path or cycle $P$ as in item (i), (ii), or (iii) from Theorem \ref{thr:characterize-popular} that is contained in $G_{N'}[S' \cup X']$. Since (by Claim \ref{cl:sameG}) $G_{N}[X\cup S]=G_{N'}[X\cup S]$ and $v \in X\cup S$ is $N$-exposed if and only if it is $N'$-exposed, we deduce $P\not\subseteq X\cup S$, else $N$ would not be $(S,X)$-locally popular, a contradiction.

        Suppose $P$ is a path: take one that is inclusionwise minimal among all paths that certify that $N'$ is not $(S',X')$-locally popular. We can then write $P$ as the juxtaposition of $(P^1,P^2)$ of levels $\ell_1,\ell_2 \neq \infty$, so that the first node of $P^2$ (which coincides with the last node of $P^1$) does not belong to $X$. Clearly $P^1,P^2 $ satisfy the hypothesis of Claim \ref{cl:fromPtoPhat} with the roles of $N'$ and $M'$ exchanged and $k=2$. Let $Q^1, Q^2$ be the paths whose existence is guaranteed by the claim. We claim that the juxtaposition of $(Q^1,Q^2)$ is a forbidden path in $G_{M'}[X'\cup S']$, a contradiction.

        First, observe that the last node of $Q^1$ coincides with the last node of $P^1$ (since the latter does not belong to $X$) which coincides with the first node of $P^2$, which coincides with the first node of $Q^2$ (again, since it does not belong to $X$). We now argue that $P^1$ and $P^2$ have no other node in common. We know they are internally vertex-disjoint and $X$-disjoint. Hence, if the first node of $P^1$ or the last node of $P^2$ belongs to $X$, we are done. Suppose not. Then, by Claim \ref{cl:fromPtoPhat}, the first node of $Q^1$ is the first node of $P^1$, which is different by hypothesis ($P$ is a path) from the last node of $P^2$, which coincides, again by Claim \ref{cl:fromPtoPhat}, with the last node of $Q^2$. Hence, the juxtaposition of $(Q^1,Q^2)$ is a path of $G$, that we denote by $Q$. Again by Claim \ref{cl:fromPtoPhat}, $Q^1$, $Q^2 \in {\cal P}_{M'}$. Moreover, the parity and level of $Q^1$ (resp. $Q^2$) coincide with the parity and level of $P^1$ (resp. $P^2$). Hence $Q \in {\cal P}_{M'}$ is the required forbidden path, obtaining the desired contradiction.

        The argument when $P$ is a cycle follows in a similar fashion, writing $P$ as the juxtaposition of $(P^1,P^2,P^3)$ with $P^1\cap X = P^2\cap X = \emptyset$, and the endpoints of $P_3$ lying in $S$ (note that we can assume that $P$ contains exactly one $(+,+)$ edge, else it also contains a path with two $(+,+)$ edges and we are back to the previous case). 
	
	\smallskip

 \noindent{\em Tipping point.}       
	Let us now show that ${\cal T'}$ is the tipping point of $N'$. First, observe that $X \subseteq X'$ implies $S'\cap X=\emptyset$. For $v \in S' \cap S$, we know that $N'(v)=N(v) = M(v)=M'(v)$. For $v \in S'\setminus S$, we have $N'(v)=M'(v)$ by construction. Hence $N'_{S'} = M'_{S'}$. Now let ${\cal C}=\{U,L,\Pi\}$ be a $(S',X')$-configuration at which $M'$ is active, with $U$ as in \eqref{eq:U}. We will show that $N'$ is also active at ${\cal C}$. A symmetric argument shows that, if ${N'}$ is active at ${\cal C}$, then so is ${M}'$, concluding the proof.

        Take paths $P^1,\dots,P^k$ that are the certificate of ${\cal C}$ at $M'$. We claim that those paths satisfy the hypothesis of Claim \ref{cl:fromPtoPhat}. They are, by construction, $S'$-paths that are pairwise $X'$-disjoint. This implies that they are pairwise internally vertex-disjoint and $X$-disjoint (using again $X\cap S'=\emptyset$). The $X'$-disjointness implies that each node from $X'$ appears at most once. Moreover, again by construction, each node from $S'$ appears at most twice. Hence, every node appears at most twice in $P^1,\dots, P^k$. Moreover, by definition of $U$, no two paths have the same endpoints. 
	
	Let $Q^1,\dots,Q^k \in {\cal P}_{N'}$ be the paths whose existence is guaranteed by Claim \ref{cl:fromPtoPhat}. We show that $Q^1,\dots,Q^k$ are the certificate of ${\cal C}$ at $N'$, concluding the proof. Fix $j \in [k]$. The parity and level of $Q^j$ is the same as that of $P^j$. Hence, the parity of $Q^j$ is $\Pi_j$, and its level $\ell_j$. Recall that $Q^j$ is an $S'$-path. By property $2$ of Claim \ref{cl:fromPtoPhat}, $Q^j$ starts (resp. ends) at the same node as $P^j$ when this belongs to $S'$, and at some node of $X'$ otherwise. Using property 1 of Claim \ref{cl:fromPtoPhat}, we deduce that $Q^j$ is a $S'$-path that starts (resp. ends) at $u_j$ (resp. $v_j$) when $u_j\neq \emptyset$ (resp. $v_j\neq \emptyset$), and at a node $x \in X'$ otherwise.

        We are left to show that paths $\{Q^j\}_{j=1,\dots,k}$ are $X'$-disjoint. Suppose not, then there exists $v \in X'$, $j\neq j'$ such that $v \in Q^j \cap Q^{j'}$. Note that $v$ cannot be an internal node of $Q^j$ or $Q^{j'}$, since those paths are pairwise internally vertex-disjoint by construction. On the other hand, if $v$ is an endpoint of both $Q^j$ and $Q^{j'}$, then the statement either follows from the $X$-disjointness of $Q^j$ and $Q^{j'}$, or by property $2$ of Claim \ref{cl:fromPtoPhat}. \hfill \qed

\end{document}